\documentclass{article}
\usepackage[utf8]{inputenc}

\usepackage{tikz}
\usepackage{pgfplots}

\usepackage[english]{babel}
\usepackage[T1]{fontenc}
\usepackage[margin=1in]{geometry} 
\usepackage[colorlinks=true, allcolors=blue]{hyperref}

\usepackage{nicefrac}
\usepackage{amsmath}
\usepackage{amsthm}
\usepackage{amssymb}
\usepackage{mathtools}
\usepackage{graphicx}
\usepackage{subcaption}
\usepackage{amssymb}
\usepackage{bm}
\usepackage{xspace}
\usepackage{xcolor}
\usepackage{comment}
\usepackage{float} 
\usepackage{thmtools}
\usepackage{thm-restate}
\usepackage{cleveref}
\usetikzlibrary{arrows}

\newtheorem{theorem}{Theorem}[section]
\newtheorem{lemma}[theorem]{Lemma}

\newtheorem{condition}[theorem]{Condition}

\newtheorem{algorithm}[theorem]{Algorithm}

\newtheorem*{theorem*}{Theorem}
\newtheorem*{problem*}{Problem}

\usepackage{caption}
\usepackage{subcaption}

\newcommand\event[1]{\mathop{\mathcal{X}\left(#1\right)}}

\newcommand\E[1]{\mathop{\mathbb{E}\left[#1\right]}}
\newcommand\Ex[2]{\mathop{\underset{#1}{\mathbb{E}}\left[#2\right]}}

\newcommand\Pro[1]{\mathop{\mathbb{P}\left[#1\right]}}

\usepackage[algo2e,ruled,noend]{algorithm2e} 

\SetCommentSty{mycommfont}

\crefname{algocfline}{Algorithm}{Algorithms}
\Crefname{algocfline}{Algorithm}{Algorithms}

\crefname{condition}{Condition}{Conditions}
\Crefname{condition}{Condition}{Conditions}

\title{MNL-Prophet: Sequential Assortment Selection under Uncertainty}

\author{Vineet Goyal\\ 
Columbia University\\
\texttt{vgoyal@ieor.columbia.edu} \and
Salal Humair \\ Amazon Science \\ \texttt{salal@amazon.com} \and 
Orestis Papadigenopoulos\\ 
Columbia University\\ 
\texttt{opapadig@columbia.edu} \and
Assaf Zeevi\\ 
Columbia University\\
\texttt{assaf@gsb.columbia.edu}
}

\date{}

\begin{document}

\maketitle

\begin{abstract}
Due to numerous applications in retail and (online) advertising the problem of assortment selection has been widely studied under many combinations of discrete choice models and feasibility constraints. In many situations, however, an assortment of products has to be constructed gradually and without accurate knowledge of all possible alternatives; in such cases, existing offline approaches become inapplicable.
We consider a stochastic variant of the assortment selection problem, where the parameters that determine the revenue and (relative) demand of each item are jointly drawn from some known item-specific distribution. The items are observed sequentially in an arbitrary and unknown order; upon observing the realized parameters of each item, the decision-maker decides irrevocably whether to include it in the constructed assortment, or forfeit it forever. The objective is to maximize the expected total revenue of the constructed assortment, relative to that of an offline algorithm which foresees all the parameter realizations and computes the optimal assortment. We provide simple threshold-based online policies for the unconstrained and cardinality-constrained versions of the problem under a natural class of substitutable choice models; as we show, our policies are (worst-case) optimal under the celebrated Multinomial Logit choice model. We extend our results to the case of knapsack constraints and discuss interesting connections to the Prophet Inequality problem, which is already subsumed by our setting. 

\end{abstract}

\section{Introduction}

The study of discrete choice models has deep roots in the theory and practice of economics and operations research -- surfacing from the need to capture, comprehend, and forecast the substitution behavior of humans; in other words, to explain the underlying process involved when people choose one among a finite set of alternatives. 
Due to its ubiquitous applications in marketing, transportation, and policy-making, discrete choice modeling has been a topic of active research for more than half a century -- the spark of numerous lines of work including axiomatic justification, empirical validation, estimation, and prediction, among others. 

In the area of revenue management, a core application of discrete choice modeling is that of assortment selection: the problem of deciding on a collection (or assortment) of products to be offered to a group of customers in order to maximize the expected revenue of the seller. More formally, let ${N}$ be a set of possible products (or items), where each $i \in {N}$ is associated with a specific revenue $r_i \in [0, \infty)$ that is collected by the seller in case of purchase. Let $\phi(i,S)$ denote the probability that a customer chooses product $i \in S$, when offered an assortment $S$ of products; naturally, $\phi(i,S)$ is identically zero for any $i \notin S$. For a given family of feasible assortments $\mathcal{F} \subseteq 2^{[n]}$, the assortment selection problem can be described as 
$$
\max_{S \in \mathcal{F}} \sum_{i \in S} \phi(i,S) \cdot r_i.
$$
In the above generic formulation, the role of discrete choice modeling is to specify the selection probabilities $\phi(i,S)$ and to capture, in that way, the consumers' demand for a particular product in the assortment. Notice, of course, that the problem may easily become intractable for most combinations of combinatorial constraints and choice probabilities, even when both have succinct representations. Nevertheless, the immense need for capturing consumers' demand in retail, travel, and (online) advertising industries motivates the research on various special instances -- ones which strike the right balance between modeling capability and tractability.

One of the most studied and well-understood threads of assortment selection is when the choice probabilities follow the celebrated Multinomial Logit (MNL) model. In the MNL model, each item $i \in N$ is associated with a non-negative demand parameter $v_i \in [0, \infty)$, called attraction\footnote{We remark that, although the MNL model can be axiomatically derived through utility theory, the attraction should not be confused with the buyer's utility for the same item; see \Cref{sec:overview} for more details.}. For any assortment $S \subseteq N$, the probability of choosing $i \in S$ is then defined as $\phi(i,S) = \frac{v_i}{v_0 + \sum_{i \in S} v_i}$, where $v_0 \in [0, \infty)$ denotes the attraction of the outside (or no-purchase) option, that is, the event where the customer does not select any of the given alternatives. The problem of assortment selection under the MNL model has been widely studied and algorithms have been derived for several combinatorial constraints of interest. This line of work includes an optimal polynomial-time algorithm for the unconstrained setting -- due to the seminal work of Talluri and van Ryzin \cite{talluri2004revenue} -- as well as optimal or near-optimal algorithms for various feasibility classes including cardinality \cite{rusmevichientong2010dynamic}, knapsack \cite{desir2022capacitated}, and totally unimodular constraints (subsuming matroids) \cite{rusmevichientong2009ptas, davis2013assortment, avadhanula2016tightness}. The problem has also been studied under generalizations of the MNL model including Nested Logit (NL) \cite{davis2014assortment,gallego2014constrained}, Markov Chain (MC) \cite{blanchet2016markov}, Generalized Attraction (GAM) \cite{gallego2015general}, and mixtures of MNL (MMNL) \cite{bront2009column,feldman2015bounding,rusmevichientong2014assortment}, to name a few. As it is natural, any reasonable (near)-optimal algorithm for the assortment selection problem relies on the knowledge of the underlying model; this includes the items' revenues and the (succinct) description of all choice probabilities.

\paragraph{Sequential assortment selection under uncertainty.} 

In many practical applications of assortment selection complete prior knowledge of the model cannot be assumed; these correspond to situations where the decision-maker (e.g., retail or advertising company) observes the characteristics of each product in a sequential manner, having little or no control over the order these are inspected. Very often, the decision of including or not a specific product in the assortment has to be made within a short period of time. In such cases, decisions can only be based on predictions of the characteristics of any potential future alternative \cite{kok2007demand} (e.g., a better wholesale price for the same product); in this regime, any known ``offline'' algorithm for assortment selection becomes inapplicable.

As an example, consider an online retail store which -- prior to the beginning of a new season (year or quarter) -- selects a collection of items to display to a particular group of customers; this can be an assortment of cell phones, clothing, household appliances etc. The retail store explores the potential products via different wholesalers, who contact and inform the store about their new arrivals in various ways (e.g., expos, demos, or brochures). After observing the features and wholesale price of each product, the retailer has to decide on whether or not merchandizing it would benefit (or damage) its overall revenue. It is often the case, however, that such a decision must be made within a short time interval; during this interval a better deal on the wholesale price could potentially be achieved (via a wholesale contract) or the trading right could be obtained by a competitor, in the case of exclusive products. Critically, the decision should be made without having knowledge of other -- potentially more profitable -- items that might be available in the near future. Similar examples can be found in online travel agencies, (click-based) advertising companies, and more.


Motivated by scenarios as the above, we formulate the {\em sequential assortment selection} problem, a variant where the decision-maker only assumes distributional knowledge of the parameters of the underlying model. Specifically, we assume that the revenue and the relative demand (see \Cref{sec:preliminaries}) of each product are drawn independently from some known joint distribution, and they are revealed to the decision-maker in an arbitrary and unknown order. Upon observing the realized parameters of each item, the decision-maker decides irrevocably whether or not to include the item in the constructed assortment. The objective becomes to maximize the expected revenue of the resulting assortment at the end of the process (over the randomness of both realized parameters and consumer's choice). For the particular case of MNL demand, for instance, we assume that the revenue $r_i$ and attraction $v_i$ of each item $i \in N$ are drawn from a known distribution $\mathcal{D}_i$ (jointly yet independently of other items).

\paragraph{The Prophet Inequality benchmark.}
In a simplistic setting where the arrival order of the items is known to the decision-maker, the optimal online policy for the sequential assortment selection problem could potentially be computed via backwards induction and dynamic programming. This approach would also require the parameter distributions to have a discrete and compact form and can be sensitive to minor changes in the input (e.g, swapping the order of two items). In our setting, however, the critical assumption of unknown arrival order makes the computation of such a policy intractable. Instead, we choose the Prophet Inequality benchmark \cite{KS77,KS78}: to compare the expected revenue of a policy with that of an omniscient and computationally powerful prophet, namely, an offline algorithm who foresees the realizations of all items and computes the optimal assortment. This is the strongest benchmark we can use for our stochastic problem, yet it enables us to obtain (worst-case) optimal guarantees. In addition, it allows us to immediately translate any (positive or negative) results into conclusions regarding the value of information -- the overall loss in expected revenue due to uncertainty. Note that, even if one could compute (or closely approximate) the optimal online policy, that would not provide any such useful insights.

\subsection{Summary of Results and Contribution}

The conceptual and technical contribution of this work to the existing literature extends in two dimensions. In the area of revenue management, we provide the first worst-case optimal policies for sequential assortment selection under uncertainty. As our results reveal -- assuming distributional knowledge of the model -- not only uncertainty can be manageable, but also has limited effect for several parametric discrete choice models (including the well-celebrated MNL demand) and standard feasibility constraints. The second dimension of our contribution is in the area of prophet inequalities, where the assortment objective appears to be an interesting novel departure from the existing work, thus yielding new challenges and insights. 

\paragraph{Outline of results.} 
In \Cref{sec:unconstrained}, we consider the unconstrained setting of the sequential assortment selection problem under a class of substitutable\footnote{Informally, in a substitutable choice model the addition of an item to an assortment can only decrease the market share (i.e., choice probability) of the already existing alternatives.} choice models (including the MNL) which satisfy two natural conditions. We design a $(1+ \gamma)$-competitive 
policy, where $\gamma \in [0,1]$ is the expected probability of purchasing (any item) in the optimal assortment; in other words, our policy collects -- in expectation -- at least a $\frac{1}{1+\gamma}$-fraction of the expected optimal revenue. We remark that $\gamma$ is an intrinsic characteristic of each instance and settings where $\gamma < 1$ (e.g., a fraction of customers never purchase any product) are distinctive for the applications of assortment selection. Our policy is a particularly simple threshold-based rule: it collects all items with realized revenue over a fixed threshold. Interestingly, this makes it oblivious to any realized parameters related to the demand (yet their distributions are used for computing the threshold).  

In \Cref{sec:cardinality}, we turn our attention to the cardinality-constrained setting, where any feasible assortment must contain at most $k$ items, for some given integer $k$. For the case of MNL demand, we design a $2$-competitive policy, which simply accepts the first $k$ items (if any) that satisfy a specific threshold condition. As opposed to the unconstrained setting, here the threshold is compared with a complex function of the realized revenue and demand parameter(s) of each item. We extend the above result to a wider class of choice models (satisfying certain conditions) and prove that a slight variation of our policy is $3$-competitive in that case. We remark that our policies (and guarantees) are robust against an arrival order that is constructed by an adversary who can foresee all the realizations and adapt to a policy's decisions. 

In \Cref{sec:lowerbounds}, we provide matching lower bounds on the competitive guarantees of our unconstrained and cardinality-constraint policies in the case of MNL demand, thus proving that our results there are (worst-case) optimal. In particular, we show that the unconstrained $(1+\gamma)$-competitive policy is optimal for any $\gamma \in [0,1]$, and that for $\gamma = 1$, the problem generalizes the original Prophet Inequality setting \cite{KS77,KS78}, where the decision-maker chooses one among $n$ random rewards. This implies that the unconstrained version of our problem indeed becomes easier when consumers are more reluctant towards purchasing.

In \Cref{sec:knapsack}, we extend our results to the case of knapsack constraints and MNL demand: each item has a random size (possibly correlated with its revenue and demand) and the total size of an assortment cannot exceed a given budget. For this setting, we provide a $\frac{2-\beta}{1-\beta}$-competitive policy, where $\beta \in (0,1]$ is an upper bound on the ratio between the realized size of any item and the budget. For instances where $\beta$ is large (and, hence, the above guarantee becomes meaningless), we provide a $5$-competitive randomized policy by leveraging the fact that the assortment objective function under any substitutable choice model is subadditive. As in the cardinality-constrainted case, our policies can be applied to a wider class of demand models (with slightly worse guarantees).

Finally, in \Cref{sec:applications}, we provide examples of parametric discrete choice models, which satisfy our sufficient conditions and, hence, to which our policies apply directly. These include a sequential adaptation of the GA model \cite{gallego2015general} and a class of choice models which can be derived through the random utility model (RUM) \cite{mcfadden1973conditional}. Further, we discuss implementation aspects of our policies and show that the provided guarantees are robust to small perturbations of our thresholds; this allows us to extend our results in the regime where the decision-maker has limited information on the priors through samples or when the corresponding offline problem is computationally hard.

\paragraph{Relation to the original Prophet Inequality.} 
The non-monotone and non-linear (and non-additively-separable) nature of the assortment objective, combined with the fact that we allow the random parameters of each item to be arbitrarily correlated, already sets us apart from the majority of existing work on prophet inequalities. However, a reader familiar with the original Prophet Inequality literature \cite{KS77,KS78,S84} may have already noticed an analogy between our proposed threshold-based policy for the unconstrained case and the (most-known) optimal policy for the latter; as our results highlight, this is not a coincidence.


As an artifact of our work, we discover several parallels and discrepancies between our setting and the original Prophet Inequality. 
To this end, our analysis and results in \Cref{sec:unconstrained,sec:lowerbounds} are revealing: for a class of discrete choice models (including MNL), the unconstrained sequential assortment selection problem can be thought of as a ``smoothed'' version of the original Prophet Inequality. In fact, our proposed policy naturally interpolates between the latter (for instances where $\gamma = 1$) and the trivial setting where the consumer never purchases any item (when $\gamma = 0$), in which case our policy becomes $1$-competitive. Indeed, the optimal (offline) solution in the second case is to add every product in the assortment and, hence, prior knowledge of the revenue and demand parameters is not particularly useful. Interestingly, one of the sufficient conditions we identify is that the underlying choice model satisfies some sort of ``negative-correlation'' among different preferences. This seems to be very relevant given that negative dependence among the rewards is the only known type of correlation under which the original Prophet Inequality result remains robust \cite{rinott1992optimal,immorlica2020prophet}.

The cardinality-constrained version of sequential assortment selection, on the other hand, presents a fundamental difference compared to the (choose-$k$) Prophet Inequality setting \cite{hajiaghayi2007automated}. In our case, the optimal competitive guarantee stays unaffected by the number $k$ of items that can be collected, as opposed to the original setting, where gradually improved guarantees can be achieved as $k$ increases \cite{alaei2014bayesian,jiang2022tight}. This implies that the hardness of our problem is mainly hidden in the non-monotonicity of our objective, rather than in the cardinality constraints.

We elaborate more on the above connections in the following section. 

\subsection{Overview of Related Work}
\label{sec:overview}
In this section, we discuss prior work in the areas of discrete choice modeling, (dynamic) assortment optimization, and prophet inequalities.

\paragraph{Discrete choice modeling.} 
The Multinomial Logit (MNL) choice model was first introduced independently by Luce \cite{luce1959individual} and Plackett \cite{plackett1975analysis}, based on natural axioms on the rationality of individual choice behavior. The model is also known as the basic attraction model (BAM), while its most common name owes to its derivation through the random utility model (RUM), introduced by McFadden \cite{mcfadden1973conditional}. In the RUM, a random utility for each alternative is drawn from a distribution, and the consumer selects the option with the highest realization. Specifically, MNL can be derived from RUM by assuming that the utility of each alternative (including the one of the outside option) is drawn from an independent Gumbel distribution. The MNL model satisfies -- by construction -- the independence of irrelevant alternatives (IIA) axiom, which suggests that the relative preference between two alternatives stays unaffected by the addition of a third (an ``irrelevant''). Despite the fact that the validity of the IIA axiom has been criticized and disproved in many settings (see red-bus/blue-bus paradox) \cite{debreu1960review,tversky1972elimination,ben1999discrete}, the MNL model remains the gold standard for choice modeling due to its simplicity and tractability, while extensions of the model have been proposed that eschew the IIA assumption \cite{simonson1989choice,mcfadden2000mixed,blanchet2016markov}. We address the reader to \cite{train2009discrete,ben1985discrete,mcfadden1977modelling} for a survey of the most commonly used discrete choice models.

\paragraph{Offline assortment selection under MNL demand.} In their seminal work, Talluri and van Ryzin \cite{talluri2004revenue} identify sufficient conditions under which any discrete choice model satisfies the so-called ``nesting-by-fare-order'' property. Specifically, this property suggests that the optimal unconstrained solution contains exactly the subset of items with revenue above a given threshold. As they authors prove, the MNL choice model satisfies this property and, hence, an optimal unconstrained solution in this case can be computed by simply trying all the subsets of this form. Interestingly, for the particular case of MNL, a choice of threshold that is guaranteed to recover the optimal assortment is the value of this assortment itself. Rusmevichientong et al. \cite{rusmevichientong2010dynamic} consider the problem of assortment selection under MNL demand in the case of cardinality constraints. By analyzing the structural properties of optimal solutions and using an interesting enumeration technique, they provide a polynomial-time algorithm for the case of cardinality constraints. Finally, Désir et al. \cite{desir2022capacitated} develop a fully polynomial-time approximation scheme (FPTAS) for knapsack-constrained assortment selection under many common choice models, including the MNL.

\paragraph{Dynamic assortment optimization and planning.} As we have already discussed, a number of studies focuses on the offline (single-leg) assortment optimization problem under various combinations of choice models and feasibility constraints. Our work is motivated by the industry practice that -- in the presence of uncertainty -- assortment selection is based on predictions of product characteristics (revenue and demand) using historical data \cite{kok2007demand,kok2015assortment}. This practice also motivates the study of the assortment optimization problem in environments where the seller repeatedly decides on a revenue-maximizing assortment of goods, based on an initially unknown (or misspecified) demand model \cite{caro2007dynamic,rusmevichientong2010dynamic,saure2013optimal,agrawal2019mnl,agrawal2017thompson}. In this line of work, the decision-maker has a dual objective: to explore the model parameters by observing the (random) consumer's selection for each offered assortment, while at the same time to maximize the total (expected) revenue collected. Following the standard regime in online learning and multi-armed bandits \cite{LR85}, in the above line of work the (additive) loss in total revenue due to uncertainty is captured by the notion of regret. 

Focusing on other sequential variants of the problem, a great deal of work considers the case of dynamic (online) assortment planning and inventory management. This refers to the setting where a seller faces an arriving stream of different customers and has to decide on an assortment of products to offer at each time, taking into account the changing stock levels and availability of products in the inventory. The customers are partitioned into types, each associated with a different demand model, and the type of each arriving customer is revealed to the decision-maker in an online manner. The main challenge here is to maximize the total (expected) revenue collected over a time horizon, without having knowledge of future customer arrivals. This line of research consists of two main threads related to the way the arrival sequence of customers is constructed: the cases of adversarial \cite{talluri2004revenue,goyal2020asymptotically,golrezaei2014real,feng2022near} and stochastic \cite{chan2009stochastic,goyal2016near} arrivals. 

In a recent work, Gallego and Berbeglia \cite{gallego2021bounds} employ prophet inequalities to study the limits of assortment personalization. In particular, they study and bound the loss due to offering a static assortment to all customers, relative to that of a clairvoyant firm which only offers the product of highest revenue among the ones that each customer is willing to purchase. We remark that our setting is very different compared to \cite{gallego2021bounds}, since we deal with a different level of randomness. More specifically, the uncertainty in our setting is due to the randomness in the model itself (including the items' revenues) and not in the consumer's final choice, which in our case is already encapsulated in the assortment objective.

\paragraph{Prophet inequalities.} 
Krengel, Sucheston, and Garling \cite{KS77,KS78} first consider the original Prophet Inequality setting, where the decision-maker observes a sequence of non-negative stochastic rewards $w_1, \ldots, w_n$. Upon observing the realization of each $w_i$, the decision-maker chooses between collecting the reward and halt, or rejecting the reward and moving to the next. For the above setting, the authors prove one of the most prominent results in optimal-stopping: assuming prior knowledge of the reward distributions, there exists a stopping-time $s$ guaranteeing that $\E{w_s} \geq \nicefrac{1}{2} \E{w_{\max}}$, where $w_{\max} = \max\{w_1, w_2, \ldots, w_n\}$ is the maximum reward. This is known to be best achievable worst-case guarantee for the above setting. Samuel-Cahn \cite{S84} later shows that the same guarantee can be achieved by a surprisingly simple policy, which first computes a threshold $\tau$ as a function of the distributions and, then, accepts the first reward (if any) that surpasses this threshold. In particular, she shows that a worst-case optimal policy can be achieved by choosing threshold $\tau = \text{median}(w_{\max})$. A few decades later, Kleinberg and Weinberg \cite{KW12} prove that the optimal guarantee can be also achieved by using (a different) threshold $\tau = {\E{w_{\max}}}/{2}$. We address the reader to \cite{CFHOV19,L17,HK92} for a survey of exiting work on prophet inequalities.

The sequential assortment selection problem is a significant departure from the original Prophet Inequality (described above) and involves several elements which could potentially preclude the design of near-optimal policies. To begin, the assortment objective does not exhibit -- at a first glance -- any particularly strong high-level properties; it is a non-linear and non-monotone set function, which becomes subadditive under substitutable demand models. Further, the addition of an item to an assortment affects the ``contributions'' (meaning the terms $\phi(i,S) \cdot r_i$) of the rest of the included items ``numerically'', through the choice probability $\phi$. Note that this holds even for the unconstrained setting of the problem and is orthogonal to presence of feasibility constraints. The above elements already separate our setting from most studied extensions of the Prophet Inequality, where the objective is either linear or additively-separable, as for example feasibility-constrained settings \cite{KW12,alaei2014bayesian,chawla2010multi,dutting2015polymatroid,rubinstein2016beyond,hajiaghayi2007automated} or combinatorial auctions \cite{correa2023constant,dutting2020prophet,feldman2014combinatorial}. 

In particular, Rubinstein and Singla \cite{rubinstein2017combinatorial} introduce the class of ``combinatorial'' prophet inequalities, where they abstract the original setting to general set functions; our maximization objective can be thought of as an instance of this class (modulo some technicalities related to the correlations between random variables). Existing works \cite{rubinstein2017combinatorial,chekuri2021submodular} develop ``combinatorial'' prophet inequalities for the cases of (non-monotone) submodular and monotone subadditive set functions, mostly relying on Online Contention Resolution Schemes (OCRS) \cite{feldman2016online}. Unfortunately, our objective does not strictly fall into any of these classes (see \cite{udwani21} for details on the discrete ``smoothness'' of assortment objectives) -- yet even if it did, the existing guarantees are either large constants or logarithmic, which is natural given the generality of the setting. More importantly, despite the elegance of the framework developed in \cite{rubinstein2017combinatorial,chekuri2021submodular}, the proposed methods are -- either computationally or literally -- hard to implement.

A second distinguishing element in our model is the presence of arbitrary correlations between the revenue and demand parameters of each item -- especially when combined with the fact that the contribution of each item to the objective depends numerically on that of different items (via the choice probabilities). It is known that that the existence of correlations among rewards can be detrimental for the Prophet Inequality problem; in fact it is impossible to achieve any non-trivial guarantees, if one allows arbitrary correlations across the rewards \cite{HK92}. In particular, Immorlica et al. \cite{immorlica2020prophet} fully characterize a class of ``linear'' correlation structures and show that the problem indeed becomes harder as the ``depth'' of their structures increases. To the best of our knowledge, the only sort of correlation under which the original Prophet Inequality result holds is that of negative-association \cite{rinott1992optimal,immorlica2020prophet}.

\section{Preliminaries}
\label{sec:preliminaries}
\subsection{Problem Definition}
Let $N = [n]$ be a set of $n$ items, where each $i \in N$ is associated with a revenue $r_i \in [0, \infty)$. Under any discrete choice model, the expected revenue of the seller for any offered assortment $S \subseteq N$ is given by 
\begin{align*}
    f(S) = \sum_{i \in S} \phi(i,S) \cdot r_i,
\end{align*}
with $\phi(i,S)$ being the probability that a customer purchases item $i \in S$ among the offered items. Similarly, $\phi(0, S)$ denotes the probability of the outside option, namely, the event where the customer does not select any of the items offered in $S$. By overloading the notation, we denote by $\phi(T,S)$ the probability that the customer purchases any item in $T \subseteq S$, when offered an assortment $S$. Notice that for any $S$, it holds $\sum_{i \in S} \phi(i,S) = 1 - \phi(0,S)$, given that the customer purchases at most one item. For any assortment $S$, let us denote by $\psi(S) = \phi(S,S)$ the probability that the customer purchases any product from $S$. For the rest of this text, we refer to the function $f : 2^N \rightarrow [0, \infty)$ as the ``total revenue'' of a solution, to avoid any confusion related to the fact that $f(\cdot)$ is itself an expectation (over customers' choice).

\paragraph{Parametric demand models.} We consider the class of parametric discrete choice (or, simply, ``demand'') models, namely, these where the choice probabilities can be succinctly described as follows: each item $i \in N$ is associated with a demand parameter $d_i$, in addition to its revenue. Depending on the demand model, these parameters can take the form of scalars, vectors, sets, or even distributions. In a parametric demand model, for any assortment $S$ and item $i \in S$, there exists a known function $\phi_{i,S}$ such that $\phi(i,S) = \phi_{i,S}(d_1, \ldots, d_n)$, that is, any choice probability can be computed as a function of the demand parameters of all items. 

\paragraph{Sequential assortment selection.} In our sequential setting, the revenue and demand parameter of every item $i \in N$ are jointly drawn from some distribution $\mathcal{D}_i$, that is $(r_i, d_i) \sim \mathcal{D}_i$, independently of other items; with a slight abuse of terminology, we refer to the pair $(r_i, d_i)$ as the ``realization'' of $i$. We remark that although we require independence across the parameters of all items, we allow the random variables $r_i$ and $d_i$ for each item $i \in N$ to be arbitrarily correlated. The decision-maker observes the (parameter) realizations of the items in an arbitrary and unknown order and -- upon observing each item -- decides irrevocably whether to include it in the constructed assortment, or skip it and move to the next. 

We focus on the class of online policies which -- given prior knowledge of the distribution $\mathcal{D}_i$ for every $i \in N$ -- adaptively construct a feasible assortment $A \in \mathcal{F}$, where $\mathcal{F} \subseteq 2^N$ is a given family of feasible assortments. The objective is to maximize the expected total revenue of the collected assortment, over the randomness of the realized parameters, relative to that of an offline algorithm, which foresees the realizations of all items and chooses the optimal assortment. Formally, we seek to find the minimum possible $\rho \in [1, \infty)$, under which an online policy satisfies a {\em prophet inequality} of the form
$$
\Ex{}{f(A)} \geq \frac{1}{\rho} \cdot \Ex{}{\max_{S \in \mathcal{F}} f(S)}.
$$
Any policy satisfying the above inequality is called $\rho$-competitive, with $\rho$ being its competitive guarantee. 

\paragraph{Adversarial arrival.} 
The online policies we propose in this work maintain their competitive guarantees even against an arrival order of the items that is constructed by an almighty adversary: an omniscient and fully-adaptive adversary, who observes the realizations of all items a priori and can adapt to a policy's decisions. 
We remark that this is the strongest possible adversary one can expect for this class of problems.

\paragraph{Notation.} 
For any set $S \subseteq N$, we denote by $\mathcal{D}[S] = \bigtimes_{i \in S} \mathcal{D}_i$ the product of the parameter distributions (and the generated randomness) of the items contained in $S$. For brevity, we define $\mathcal{D} = \mathcal{D}[N]$ as the product of the distributions of all items. We denote by $\event{\mathcal{E}}\in \{0,1\}$ the indicator function which yields $1$ if and only if the event $\mathcal{E}$ holds.
We define $S^* = \arg\max_{S \in \mathcal{F}} f(S)$ to be the optimal feasible subset for a given instance (that is, a parametric choice model and realized parameters $(r_i,d_i)$ for all $i \in N$). For ease of notation, we suppress any dependence of $S^*$ on $f$ and $\mathcal{F}$, since it will always be clear from the context. We use the standard notation $[n] = \{1,2,\ldots,n\}$ for any integer $n$ and $(x)^+ = \max(x , 0)$ to denote the maximum between a real number and zero. Finally, for any set $S \subseteq [n]$ and vector $v \in [0,\infty)^{n}$, we denote $v(S) = \sum_{i \in S} v_i$.

\subsection{Multinomial Logit Choice Model and Sufficient Conditions}

The results of this work apply to any parametric demand model that satisfies certain sufficient conditions, with the most important example of this class being the well-known MNL choice model. In this paragraph, we formally define the latter and use this as a reference to describe our conditions.

\paragraph{Multinomial Logit choice model.} In the MNL choice model, each item $i \in N$ is associated with an attraction parameter $v_i \in [0, \infty)$, with $v_0 \in [0, \infty)$ being the attraction of the outside option. The probability that item $i$ is selected from assortment $S \subseteq N$ is then defined as $\phi(i,S) = \frac{v_i}{v_0 + v(S)}$, when $i \in S$, and $0$, otherwise. Throughout this work, we assume that the attraction of the outside option is deterministic\footnote{Our results can be extended to the case where the outside attraction is random, yet known to the decision-maker a priori. In this case, our guarantees hold ``pointwise'' (for any $v_{0}$), by applying our policies {\em conditioned} on the realization of $v_0$.}. Clearly, MNL is a parametric choice model with the demand parameter of each item coinciding with its attraction, namely, $d_i \equiv v_i$ for all $i \in N$.

\paragraph{Sufficient Conditions.} Our results apply to substitutable choice models, that is, models where the addition of an item $i \in N \setminus S$ to an assortment $S$ does not increase the probability of choosing any of the already existing items (including that of the outside option). More formally, for any $S \subseteq T \subseteq N$ and $i \in S \cup \{0\}$, it holds $\phi(i,S) \geq \phi(i,T)$. 

The following condition characterizes the ability of the decision-maker to evaluate any choice probability $\phi(i,S)$ by using only the (observed) demand parameters of the items in $S$, together with any deterministic information. 

\begin{condition}[Independent Evaluation] \label{cond:1}
For any assortment $S \subseteq N$ and item $i \in S \cup \{0\}$, the probability $\phi(i,S)$ only depends on the demand parameters of the items in $S$.
\end{condition}

The above condition excludes a number of parametric discrete choice models, yet it seems critical for our sequential setting. Informally, it suggests that the decision-maker should be able to evaluate (without any uncertainty) the total revenue of the already collected items at any point in time. Notice that, using the independence of the demand parameters across different items, the above condition implies that for any $S,T \subseteq N$ with $S \cap T = \emptyset$ and items $i \in S \cup \{0\}, i' \in T$, it holds $\E{\phi(i,S) \cdot \phi(i',T)} = \E{\phi(i,S)} \cdot \E{\phi(i',T)}$. 

The second condition associates the probability of choosing item $i$ in an assortment $S$ with the probability that $i$ would be purchased if offered alone. Informally, it implies the existence of some sort of negative dependence among the demand of different preferences; the choice of its name is made for drawing a parallel with the Prophet Inequality setting under negatively-correlated rewards \cite{rinott1992optimal}. 

\begin{condition}[Negative Correlation] \label{cond:2}
For any assortment $S \subseteq N$ and item $i \in S$, it holds 
$$
\phi(i,S) \geq \phi(i,\{i\}) \cdot \phi(0, S \setminus \{i\}).
$$
\end{condition}

The validity of \Cref{cond:2} is crucial, since -- combined with \Cref{cond:1} -- it suffices for the design of optimal policies for the unconstrained setting under various demand models (see \Cref{sec:applications}).

The third and last condition describes an expansion property of a demand model and comes in two versions: ``strong'' and ``weak''. The strong version suggests that the ratio between the market share of the outside option and that of any other item in an assortment is unaffected by the rest of the included items; this enables the design of an optimal cardinality-constrained policy for the case of MNL demand. 

\begin{condition}[Independence of Irrelevant Alternatives] \label{cond:3}
For any set $S \subseteq N$ and $i \in S$, we have
$$
\frac{\phi(0,S)}{\phi(i,S)} = \frac{\phi(0,\{i\})}{\phi(i,\{i\})}.
$$
Further, we say that demand model satisfies the ``{\em weak version}'' of this condition, if for any $S \subseteq N$ and $i \in S$ it holds $\frac{\phi(0,S)}{\phi(i,S)} \leq \frac{\phi(0,\{i\})}{\phi(i,\{i\})}
$. 
\end{condition}
We remark that the strong version of the above condition is implied by the axiom of Independence of Irrelevant Alternatives (IIA), and MNL is essentially the only non-trivial choice model satisfying the latter. Its weak version, on the other hand, is satisfied by several demand models (see \Cref{sec:applications}) and -- when combined with \Cref{cond:1,cond:2} -- it allows the design of near-optimal policies for the case of cardinality and knapsack constraints. 

We can verify that MNL model satisfies all the aforementioned conditions (see \cref{app:demand}):

\begin{restatable}{proposition}{restateMNLfacts} 
The MNL is a substitutable discrete choice model which satisfies \Cref{cond:1,cond:2} and the strong version of \Cref{cond:3}.
\end{restatable}

\section{Optimal Fixed-Threshold Policy for the Unconstrained Setting} \label{sec:unconstrained}

In this section, we study the unconstrained setting of sequential assortment selection, where the decision-maker can collect any combination of arriving items. In particular, for any substitutable demand model satisfying \Cref{cond:1,cond:2} (that includes the MNL), we provide a $(1+\gamma)$-competitive policy, where $\gamma = \E{\psi(S^*)} = \E{\phi(S^*,S^*)}$ is the expected probability of purchasing any item in the optimal assortment. As we show in \Cref{sec:lowerbounds}, for this class of discrete choice models, the above guarantee is worst-case optimal for any $\gamma \in [0,1]$.

We propose a simple threshold-based policy which accepts all items with realized revenue greater than or equal to a fixed threshold that can be computed offline\footnote{In \Cref{sec:applications}, we discuss (efficient) implementation aspects of our policies in terms of computational and sample complexity.}. By denoting $S^*$ the optimal unconstrained assortment, our policy can be described as follows:

\begin{algorithm} \label{alg:unconstrained}
Compute threshold $\tau = \frac{1}{1+\gamma} \cdot \Ex{}{f(S^*)}$, where $\gamma = \E{\psi(S^*)}$, and set $A_{\tau} \gets \emptyset$. Accept every arriving item with revenue greater or equal to $\tau$, namely, add $i$ to $A_\tau$ if and only if $r_i \geq \tau$. 
\end{algorithm}

Remarkably, the above policy only uses the realized revenue of each arriving item to decide on whether or not to accept it; this fact makes it oblivious to any realized parameters which are directly related to the demand. 

In the following result, we bound the competitive guarantee of \Cref{alg:unconstrained}:

\begin{theorem}\label{thm:unconstrained}
For any substitutable choice model satisfying \Cref{cond:1,cond:2}, \Cref{alg:unconstrained} is a $(1+\gamma)$-competitive policy for the unconstrained setting of sequential assortment selection, where $\gamma = \E{\psi(S^*)}$ is the expected probability of purchasing any item in the optimal assortment. Moreover, by setting threshold $\tau = \frac{1}{2} \cdot \Ex{}{f(S^*)}$, \Cref{alg:unconstrained} becomes $2$-competitive (independently of $\gamma$).
\end{theorem}

Notice that the above guarantee matches the optimal guarantee of the original Prophet Inequality setting, in the case where $\gamma = 1$. However, in many real-life applications of assortment selection, it is reasonable to assume that $\gamma$ is significantly smaller. Indeed, it is easy to find examples of markets, where a constant fraction of buyers just browse and leave without purchasing any product. 

In fact, for demand models where the buyer never chooses the outside option (that is, $\phi(0,S) = 0$ for any $S \subseteq N$) and, thus, $\gamma = 1$, one can design an optimal unconstrained policy that does not rely on any assumptions (see \Cref{app:alternative} for details):

\begin{theorem}
For any demand model which satisfies $\phi(0,S) = 0$ for any $S \subseteq N$, there exists an optimal $2$-competitive fixed-threshold policy for the unconstrained setting.
\end{theorem}

The remainder of this section is dedicated to proving \Cref{thm:unconstrained}.

\subsection{Proof of \Cref{thm:unconstrained}}
Let us fix any threshold $\tau$ and let $A_\tau = \{i \in N \mid r_i \geq \tau\}$ denote the (random) subset of items with (realized) revenue at least $\tau$. Note that, in the unconstrained setting, $A_{\tau}$ coincides with the set of collected items. 

We first observe that, for any parameter realization, the total revenue of the collected items in $A_{\tau}$ can be rewritten as
\begin{align}
    f(A_{\tau}) = \sum_{i \in A_{\tau}} \phi(i, A_{\tau}) \cdot r_i = \sum_{i \in N} \phi(i, A_{\tau}) \cdot (r_i - \tau)^+ + \phi(A_\tau,A_\tau)\cdot \tau, \label{eq:unconstraineddecomposition}
\end{align}
where we use the fact that an item belongs to $A_{\tau}$ if and only if $r_i \geq \tau$ and that $\sum_{i \in S} \phi(i,S) = \phi(S,S)$, by definition of $\phi$.

Let us define constant $\lambda = \Ex{\mathcal{D}[N]}{\phi(0, A_\tau)}$ and notice that it holds $1-\lambda = \Ex{\mathcal{D}[N]}{\phi(A_{\tau}, A_\tau)}$ -- again by definition of $\phi$. 

In the next lemma we bound the expectation of the first term in the RHS of \eqref{eq:unconstraineddecomposition} over the randomness of the realized revenues and demand parameters. 

\begin{lemma}\label{lem:unconstrained:term1}
For any substitutable choice model satisfying \Cref{cond:1,cond:2}, it holds
$$
\Ex{\mathcal{D}[N]}{\sum_{i \in N} \phi(i, A_{\tau}) \cdot (r_i - \tau)^+} \geq \lambda \cdot \left(\Ex{\mathcal{D}[N]}{f(S^*)} - \Ex{\mathcal{D}[N]}{\phi(S^*,S^*)} \tau \right),
$$
where $S^*$ denotes the optimal unconstrained solution for a given realization.
\end{lemma}
\begin{proof}
Recall that for each item $i \in N$, the parameters $(r_i, {d}_i)\sim \mathcal{D}_i$ are drawn independently of the corresponding parameters of different items. For any subset $S \subseteq {N}$, let us denote by $\mathcal{D}[S]$ the randomness generated by the parameters of the items in $S$. By \Cref{cond:1} it becomes clear that for any set $S$ and item $i \in S$, the (random) probability of choosing item $i$ from $S$, that is $\phi(i,S)$, is measurable with respect to $\mathcal{D}[S]$ (and independent of $\mathcal{D}[N \setminus S]$).


By linearity of expectation and the fact that the parameters of each item are independent, we have 
\begin{align*}
    \Ex{\mathcal{D}[N]}{ \sum_{i \in N} \phi(i,A_\tau) \cdot (r_i - \tau)^+ } &= \sum_{i \in N} \Ex{\mathcal{D}_i}{\Ex{\mathcal{D}[{N}\setminus\{i\}]}{\phi(i,A_\tau) \cdot (r_i - \tau)^+}} \\
    &\geq \sum_{i \in N} \Ex{\mathcal{D}_i}{\Ex{\mathcal{D}[{N}\setminus\{i\}]}{\phi(i,\{i\}) \cdot \phi(0,A_\tau \setminus \{i\}) \cdot (r_i - \tau)^+ }} \\
    &= \sum_{i \in N} \Ex{\mathcal{D}_i}{\phi(i,\{i\}) \cdot (r_i - \tau)^+ \cdot  \Ex{\mathcal{D}[{N}\setminus\{i\}]}{ \phi(0,A_\tau \setminus \{i\})}} \\
     &= \sum_{i \in N} \Ex{\mathcal{D}[N \setminus \{i\}]}{ \phi(0,A_\tau \setminus \{i\})} \cdot \Ex{\mathcal{D}_i}{\phi(i,\{i\}) \cdot (r_i - \tau)^+},
\end{align*}
where the first inequality above follows from \Cref{cond:2}, which suggests that $\phi(i, A_{\tau}) \geq \phi(i, \{i\}) \cdot \phi(0, A_{\tau} \setminus \{i\})$. The second equality follows from the fact that $r_i$ and $\phi(i,\{i\})$ only depend on the realized parameters of item $i$, which follows from \Cref{cond:1}. The last equality follows from noticing that the set $A_{\tau} \setminus \{i\}$ does not depend on the realized parameters of item $i$ and, hence, neither does $\phi(0, A_\tau \setminus \{i\})$, again using \Cref{cond:1}.

For any fixed $i \in N$, by using the fact that $\phi(0,A_\tau \setminus \{i\})$ is not affected by introducing an independent copy $(r_i,d_i) \sim \mathcal{D}_i$ (by \Cref{cond:1}), we can see that
$$
\Ex{\mathcal{D}[{N} \setminus \{i\}]}{\phi(0,A_\tau \setminus \{i\})} = \Ex{\mathcal{D}[N]}{\phi(0,A_\tau \setminus \{i\})} \geq \Ex{\mathcal{D}[N]}{\phi(0,A_\tau)} = \lambda,
$$
where the inequality follows from the substitutability of the demand model, while the last equality holds by definition of $\lambda$. By combining the above bounds, we get
\begin{align*}
    \Ex{\mathcal{D}[N]}{\sum_{i \in N} \phi(i,A_\tau) \cdot (r_i - \tau)^+ } \geq \lambda \cdot \sum_{i \in N} \Ex{\mathcal{D}_i}{\phi(i,\{i\}) \cdot (r_i - \tau)^+} = \lambda \cdot \Ex{\mathcal{D}[N]}{\sum_{i \in N} \phi(i,\{i\}) \cdot (r_i - \tau)^+},
\end{align*}
where the equality above follows from coupling the (independent) parameter realizations of the items.

We can further lower bound the above expectation as follows:
\begin{align*}
    \Ex{\mathcal{D}[N]}{\sum_{i \in N} \phi(i,A_\tau) \cdot (r_i - \tau)^+ } &\geq \lambda \cdot \Ex{\mathcal{D}[N]}{\sum_{i \in N} \phi(i,\{i\}) \cdot (r_i - \tau)^+} \\
    &\geq \lambda \cdot \Ex{\mathcal{D}[N]}{\sum_{i \in S^*} \phi(i,S^*) \cdot (r_i - \tau)^+} \\
    &\geq \lambda \cdot \left( \Ex{\mathcal{D}[{N}]}{ \sum_{i \in S^*} \phi(i,S^*) \cdot r_i} - \Ex{\mathcal{D}[{N}]}{\phi(S^*,S^*)} \cdot \tau \right) \\ 
    &\geq \lambda \cdot \left( \Ex{\mathcal{D}[{N}]}{ f(S^*)} - \Ex{\mathcal{D}[{N}]}{\phi(S^*,S^*)} \cdot \tau \right),
\end{align*}
where in the second equality we use the fact that every term in the summation is non-negative, and that $\phi(i,\{i\}) \geq \phi(i,S^*)$ for each $i \in S^*$, by substitutability. In the penultimate inequality we simply use that $(x)^+ \geq x$ for any real $x$. 
\end{proof}

By taking expectation over the second term in the RHS of \eqref{eq:unconstraineddecomposition} and by definition of $\lambda$, we observe that 
\begin{align}
\Ex{\mathcal{D}[{N}]}{\phi(A_\tau,A_\tau)} \tau = (1-\lambda) \cdot \tau. \label{eq:unconstrained:term1}
\end{align}

We are now ready to complete the proof.

\begin{proof}[Proof of \Cref{thm:unconstrained}] By combining decomposition \eqref{eq:unconstraineddecomposition} with \Cref{lem:unconstrained:term1} and equation \eqref{eq:unconstrained:term1}, for the expected total revenue collected by \Cref{alg:unconstrained} using threshold $\tau$, we have
\begin{align*}
\Ex{\mathcal{D}[{N}]}{f(A_\tau)} \geq \lambda \cdot \left(\Ex{\mathcal{D}[{N}]}{f(S^*)} - \Ex{\mathcal{D}[N]}{\phi(S^*,S^*)} \cdot \tau\right) + (1 - \lambda) \cdot \tau.
\end{align*}
By setting threshold $\tau = \frac{1}{1+\gamma} \cdot \Ex{\mathcal{D}[N]}{f(S^*)}$ with $\gamma = \Ex{\mathcal{D}[N]}{\phi(S^*,S^*)}$, we can see that \Cref{alg:unconstrained} becomes $\left(1+\gamma\right)$-competitive. In addition, by setting threshold $\tau = \frac{1}{2}\Ex{\mathcal{D}[{N}]}{f(S^*)}$ and noticing that $\Ex{\mathcal{D}[N]}{\phi(S^*,S^*)} \leq 1$, the resulting policy becomes $2$-competitive. This concludes the proof. 
\end{proof}

\section{Optimal Threshold Policy for Cardinality Constraints} \label{sec:cardinality}

We consider the cardinality-constrained version of sequential assortment selection, where any feasible assortment contains at most $k$ items. In this case, we provide a $2$-competitive policy for the class of substitutable demand models satisfying \Cref{cond:1,cond:2,cond:3}, which includes the MNL; as we show in \Cref{sec:lowerbounds}, this policy is optimal for this class of discrete choice models. Further, for demand models that satisfy the weak version of \Cref{cond:3}, a slight adaptation of our policy becomes $3$-competitive.

As opposed to the unconstrained case, our policy for cardinality constraints requires to observe both the revenue and the demand parameter of each arriving item in order to make decisions. For a fixed threshold $\tau > 0$, our policy works as follows: for each arriving item $i$ with realized revenue $r_i$ and demand parameter $d_i$, the policy first checks the condition $\phi(i,\{i\}) (r_i-\tau) \geq \frac{\phi(0,\{i\})}{k} \tau$, where $k$ is the maximum number of items that can be collected. Notice that, for each $i$, the validity of this condition already guarantees that $r_i \geq \tau$, and can be tested by only observing the demand parameter $d_i$, by \Cref{cond:1}. If the condition is true, then the policy collects the item only if this does not violate the cardinality constraints given the already collected items. By setting threshold $\tau = \frac{1}{2}\E{f(S^*)}$, where $S^*$ is the optimal cardinality-constrained assortment for a given instance, the above policy becomes $2$-competitive under \Cref{cond:1,cond:2,cond:3}. 

\begin{algorithm} \label{alg:cardinality}
Compute threshold $\tau = \frac{1}{2}\Ex{}{f(S^*)}$ and set ${A}_{\tau} \gets \emptyset$. For each arriving item $i \in N$, if $\phi(i,\{i\}) (r_i-\tau) \geq \frac{\phi(0,\{i\})}{k} \cdot \tau$ and $|A_{\tau}| \leq k -1$, then add $i$ to $A_{\tau}$. 
\end{algorithm}

Interestingly, in the case of MNL demand the threshold condition is equivalent to testing $\frac{v_i \cdot r_i}{\nicefrac{v_0}{k} + v_i} \geq \tau$, that is, whether the total revenue of collecting item $i$ alone would be at least $\tau$, assuming that the outside attraction is down-scaled by a factor of $k$.

The main result of this section can be summarized as follows:

\begin{theorem}\label{thm:cardinality}
Under MNL demand (or any substitutable choice model satisfying \Cref{cond:1,cond:2,cond:3}), \Cref{alg:cardinality} is $2$-competitive for the cardinality-constrained setting of sequential assortment selection. 
\end{theorem}

By slightly modifying the above policy, one can get a $3$-competitive for a class of demand models which satisfy the weak version of \Cref{cond:3}.

\begin{restatable}{theorem}{restatecardinalityweak} \label{thm:cardinality3}
Under any substitutable choice model satisfying \Cref{cond:1,cond:2}, and the weak version of \Cref{cond:3}, by setting $\tau = \frac{1}{3}\Ex{}{f(S^*)}$, \Cref{alg:cardinality} becomes $3$-competitive for cardinality-constraints. 
\end{restatable}

We remark that, under the assumptions of \Cref{thm:cardinality3}, one can get a slightly better $(2+\gamma)$-competitive policy, where $\gamma = \E{\psi(S^*)}$ is the expected probability of purchasing (any item) in the optimal cardinality-constrained assortment, by setting threshold $\frac{1}{2+\gamma} \cdot \E{f(S^*)}$.

In the rest of this section we prove \Cref{thm:cardinality}; the proof of \Cref{thm:cardinality3} can be found in \Cref{app:omitted}.

\subsection{Proof of \Cref{thm:cardinality}}\label{sec:cardinality:proof}
Let us fix a threshold $\tau > 0$ and let $E_{\tau} = \{i \in N \mid \phi(i,\{i\}) (r_i-\tau)^+ \geq \frac{\phi(0,\{i\})}{k} \cdot \tau\}$ denote the (random) subset of items that satisfy the threshold condition of \Cref{alg:cardinality}. We remark that the addition of each item $i \in N$ to $E_{\tau}$ only depends of the parameter realization of $i$. Let $\sigma : {N} \rightarrow |{N}|$ be a arrival order of the items for a given realization and let $A^{\prec i}_\tau$ be the subset of items collected by our policy exactly before the arrival of item $i$. Notice that the sets $\{A^{\prec i}_\tau\}_{i \in N}$ depend on the chosen arrival order $\sigma$. We denote by $A_\tau \subseteq N$ the selected assortment after the arrival of all items and observe that $A_{\tau} \subseteq E_{\tau}$, since every collected item has to satisfy the threshold condition. 

At this point, it is important to make two notes regarding the case of adversarial order: the threshold-based deterministic nature of our policies does not allow an adversary to benefit from adaptivity, since the worst-case arrival order can be computed offline, after observing the parameter realizations of all items. Further, for the sake of analysis and without loss of generality we assume that, for every parameter realization, the worst-case arrival order $\sigma$ is (effectively) unique. Indeed, in the case of two (or more) equivalent orders the adversary can choose the ``smallest'' one, based on a predetermined lexicographic ordering on the items. 

For any realization of the parameters of all items, the total revenue collected by \Cref{alg:cardinality} using a fixed threshold $\tau > 0$ can be rewritten as follows:
\begin{align}
f(A_{\tau}) &= \sum_{i \in A_{\tau}} \phi(i,A_{\tau})\cdot r_i \notag\\
&= \sum_{i \in A_{\tau}} \phi(i,A_\tau) (r_i - \tau)^+ + \phi(A_{\tau},A_{\tau}) \cdot \tau \notag \\
&= \sum_{i \in A_{\tau}} \frac{\phi(i,A_{\tau})}{\phi(i,\{i\})} \phi(i,\{i\}) (r_i - \tau)^+  + \phi(A_{\tau},A_{\tau}) \cdot \tau \notag \\
&= \sum_{i \in A_{\tau}} \frac{\phi(i,A_{\tau})}{\phi(i,\{i\})} \left( \phi(i,\{i\})(r_i - \tau)^+ -\frac{\phi(0,\{i\})}{k} \tau\right)^+ + \left(\phi(A_{\tau},A_{\tau}) + \sum_{i \in A_{\tau}}\frac{\phi(0,\{i\})}{k}\frac{\phi(i,A_{\tau})}{\phi(i,\{i\})} \right) \tau, \label{eq:decompositioncardinality}
\end{align}
where in the second equality above, we are allowed to drop the negative part of $r_i-\tau$, since we know that $r_i \geq \tau$, for any $i \in A_{\tau}$, by definition of the threshold condition. We remark that, for any item $i \in A_{\tau}$, it has to be that $\phi(i,\{i\}) > 0$ since, in the opposite case, $\frac{\phi(0,\{i\})}{k} \tau$ must be strictly positive and the item cannot satisfy the threshold condition; hence, dividing by $\phi(i,\{i\})$ in the third equality above is not an issue. Notice, further that, as opposed to the (analogous) decomposition \eqref{eq:unconstraineddecomposition} for the unconstrained setting, here we cannot replace $A_{\tau}$ with $N$ in the last equality, since not every threshold-feasible item is collected.

Let us define constant $\lambda = \Ex{\mathcal{D}[N]}{\phi(0,A_{\tau}) \event{|A_{\tau}| < k}}$, where $\event{|A_{\tau}| < k} \in \{0,1\}$ is the indicator of the event $\{|A_{\tau}| < k\}$. Using that the worst-case arrival order is w.l.o.g. unique for each realization, and since $|A_{\tau}| \leq k$, we observe that $1-\lambda = \Ex{\mathcal{D}[N]}{\phi(0,A_{\tau}) \event{|A_{\tau}| = k} + \phi(A_{\tau},A_{\tau})}$. 

In the next lemma, we bound the expectation of the first term in the RHS of decomposition \eqref{eq:decompositioncardinality}. 

\begin{lemma}\label{lem:cardinality:term1}
For any substitutable demand model satisfying \Cref{cond:1,cond:2,cond:3}, it holds
$$
\Ex{\mathcal{D}[N]}{\sum_{i \in A_{\tau}} \frac{\phi(i,A_{\tau})}{\phi(i,\{i\})} \left( \phi(i,\{i\})(r_i - \tau)^+ -\frac{\phi(0,\{i\})}{k} \tau\right)^+} \geq \lambda \cdot \left(\Ex{\mathcal{D}[N]}{f(S^*)} - \tau \right),
$$
where $S^*$ denotes the optimal cardinality-constrained solution for a given realization.
\end{lemma}
\begin{proof}
Let us denote $\zeta(i) = \left( \phi(i,\{i\})(r_i - \tau)^+ -\frac{\phi(0,\{i\})}{k} \cdot \tau\right)^+$ and observe that for any item $i \in N \setminus E_{\tau}$ that does not satisfy the threshold condition of \Cref{alg:cardinality}, it must be $\zeta(i) = 0$. Recall that, by definition of the policy, an item is collected only if $i \in E_{\tau}$ and the already collected items by the time $i$ arrives are less than $k$, that is, $|A^{\prec i}_{\tau}| < k$. Using the above, we have 
\begin{align*}
\Ex{\mathcal{D}[N]}{\sum_{i \in A_{\tau}} \frac{\phi(i,A_{\tau})}{\phi(i,\{i\})} \left( \phi(i,\{i\})(r_i - \tau)^+ -\frac{\phi(0,\{i\})}{k} \tau\right)^+} &= \Ex{\mathcal{D}[N]}{\sum_{i \in N} \frac{\phi(i,A_\tau)}{\phi(i,\{i\})} \zeta(i) \event{|A^{\prec i}_{\tau}| < k}} \\
&\geq \Ex{\mathcal{D}[N]}{\sum_{i \in N} \phi\left(0,A_{\tau}\setminus \{i\}\right)  \zeta(i) \event{|A^{\prec i}_{\tau}| < k}},
\end{align*}
where the inequality follows from \Cref{cond:2}, which suggests that for any set $S$ and $i \in S$, it holds $\phi(i,S) \geq \phi(i,\{i\}) \cdot \phi(0,S\setminus \{i\})$. Note, again, that for any term that appears in the summation such that $\zeta(i) > 0$, it has to be that $\phi(i,\{i\})>0$ (hence, division by zero in the first equality is not an issue). 

The fact that the parameters of each item are realized independently combined with linearity of expectation allows us to rewrite the RHS of the above inequality as
\begin{align*}
\Ex{\mathcal{D}[N]}{\sum_{i \in N} \phi\left(0,A_{\tau}\setminus \{i\}\right) \zeta(i) \event{|A^{\prec i}_{\tau}| < k}} = \sum_{i \in N} \Ex{\mathcal{D}_i}{\zeta(i) \Ex{\mathcal{D}[N\setminus \{i\}]}{\phi\left(0,A_{\tau}\setminus \{i\}\right)  \event{|A^{\prec i}_{\tau}| < k}}}.
\end{align*}

Observe that the value of $\zeta(i)$ in the above expression only depends on the realized parameters of item $i$ (using \Cref{cond:1}), yet the random set $A_{\tau} \setminus \{i\}$ can be correlated with $\zeta(i)$ through the arrival order (which is decided after observing all the realizations); this in turn correlates the values of $\phi(0,A_{\tau} \setminus \{i\})$ and $\event{A^{\prec i}_{\tau}}$ with $\zeta(i)$ in the above expression. To overcome this issue and ``uncorrelate'' the above random quantities, we recall the definition of set $E_{\tau} = \{i \in N \mid \zeta(i)>0\}$, that is the subset of items that satisfy the threshold condition, and notice that the following inequality holds pointwise for any $i \in N$: 
\begin{align*}
\phi(0,A_{\tau} \setminus \{i\}) \event{|A^{\prec i}_{\tau}| < k} \geq \phi(0,A_{\tau} \setminus \{i\}) \event{|A_{\tau} \setminus \{i\}| < k} = \phi(0,E_{\tau} \setminus \{i\})\event{|E_{\tau} \setminus \{i\}| < k}.
\end{align*}
The above inequality follows from the fact that the event $|A_{\tau} \setminus \{i\}| < k$ implies that $|A^{\prec i}_{\tau}| < k$, independently of the order $i$ arrives. The equality follows from observing that whenever the event $|E_{\tau} \setminus \{i\}| < k$ holds true, then the sets $A_{\tau} \setminus \{i\}$ and $E_{\tau} \setminus \{i\}$ must coincide. In other words, if strictly less than $k$ items from $N \setminus \{i\}$ satisfy the threshold, then these items must be collected by the policy. By combining this with the fact that the set $E_{\tau}$ does not depend on the arrival order (and hence $E_{\tau} \setminus \{i\}$ is independent of $\zeta(i)$), we get
\begin{align*}
\sum_{i \in N} \Ex{\mathcal{D}_i}{\zeta(i) \Ex{\mathcal{D}[N\setminus \{i\}]}{\phi\left(0,A_{\tau}\setminus \{i\}\right)  \event{|A^{\prec i}_{\tau}| < k}}} &\geq \sum_{i \in N} \Ex{\mathcal{D}_i}{\zeta(i) \Ex{\mathcal{D}[N\setminus \{i\}]}{\phi\left(0,E_{\tau}\setminus \{i\}\right)  \event{|E_{\tau} \setminus \{i\}| < k}}} \\
&= \sum_{i \in N} \Ex{\mathcal{D}_i}{\zeta(i)} \Ex{\mathcal{D}[N\setminus \{i\}]}{\phi\left(0,E_{\tau}\setminus \{i\}\right)  \event{|E_{\tau} \setminus \{i\}| < k}}, 
\end{align*}
where, in the equality, we use the fact that $E_{\tau} \setminus \{i\}$ is independent of $i$ (and, hence, of $\zeta(i)$) and that the value of $\phi\left(0,E_{\tau}\setminus \{i\}\right)$ only depends on the parameters of the items in $E_{\tau}\setminus \{i\}$, by \Cref{cond:1}.

Notice that introducing an independent copy $(r_i,d_i) \sim \mathcal{D}_i$ for item $i$ does not affect the value of the right expectation above, again by \Cref{cond:1}. Thus, for any $i \in N$, we get
\begin{align*}
\Ex{\mathcal{D}[N\setminus \{i\}]}{\phi\left(0,E_{\tau}\setminus \{i\}\right)  \event{|E_{\tau} \setminus \{i\}| < k}} &= \Ex{\mathcal{D}[N]}{\phi\left(0,E_{\tau}\setminus \{i\}\right)  \event{|E_{\tau} \setminus \{i\}| < k}} \\
&\geq \Ex{\mathcal{D}[N]}{\phi\left(0,E_{\tau}\right)  \event{|E_{\tau}| < k}} \\
&= \Ex{\mathcal{D}[N]}{\phi\left(0,A_{\tau}\right)  \event{|A_{\tau}| < k}} \\
&= \lambda,
\end{align*}
where the first inequality follows from substitutability and the fact that $\{|E_\tau| < k \} \subseteq \{|E_\tau \setminus \{i\}| < k \}$. The last two equalities hold since whenever at most $k-1$ items pass the threshold the sets $E_\tau$ and $A_\tau$ must coincide, and by definition of $\lambda$. 

By combining the above elements and recalling that $\zeta(i) = \left( \phi(i,\{i\})(r_i - \tau)^+ -\frac{\phi(0,\{i\})}{k} \cdot \tau\right)^+ \geq 0$ for each $i \in N$ is independent of the realizations of the rest of the items, we reach that 
\begin{align*} 
\Ex{\mathcal{D}[N]}{\sum_{i \in A_{\tau}} \frac{\phi(i,A_{\tau})}{\phi(i,\{i\})} \left( \phi(i,\{i\})(r_i - \tau)^+ -\frac{\phi(0,\{i\})}{k} \tau\right)^+} \geq \lambda \cdot \Ex{\mathcal{D}[N]}{\sum_{i \in N} \left( \phi(i,\{i\})(r_i - \tau)^+ -\frac{\phi(0,\{i\})}{k} \cdot \tau\right)^+}.
\end{align*}

Let $S^*$ be the optimal cardinality-constrained solution and observe that, for any item $i \in S^*$, without loss of generality it holds $\phi(i,\{i\}) > 0$. Indeed, in the opposite case, substitutability implies that $\phi(i,S^*) = 0$ and, for any $i' \in S^* \setminus \{i\}$ it must be $\phi(i', S^* \setminus \{i\}) \geq \phi(i', S^*)$. Therefore, in the case where $\phi(i,\{i\}) = 0$, item $i$ can be removed from the optimal solution without damaging the total revenue. Hence, we have
\begin{align*} 
\Ex{\mathcal{D}[N]}{\sum_{i \in A_{\tau}} \frac{\phi(i,A_{\tau})}{\phi(i,\{i\})} \left( \phi(i,\{i\})(r_i - \tau)^+ -\frac{\phi(0,\{i\})}{k} \tau\right)^+} &\geq \lambda \Ex{\mathcal{D}[N]}{\sum_{i \in S^*} \left( \phi(i,\{i\})(r_i - \tau)^+ -\frac{\phi(0,\{i\})}{k} \tau\right)^+} \notag \\
&= \lambda \Ex{\mathcal{D}[N]}{\sum_{i \in S^*} \phi(i,\{i\})\left((r_i - \tau)^+ -\frac{\phi(0,\{i\})}{\phi(i,\{i\})} \frac{\tau}{k} \right)^+} \notag\\
&\geq \lambda \Ex{\mathcal{D}[N]}{\sum_{i \in S^*} \phi(i,S^*)\left((r_i - \tau)^+ -\frac{\phi(0,S^*)}{\phi(i,S^*)} \frac{\tau}{k}\right)^+} \notag\\
&\geq \lambda \Ex{\mathcal{D}[N]}{\sum_{i \in S^*} \phi(i,S^*) \cdot r_i - \left(\phi(S^*,S^*)  + \phi(0,S^*) \frac{|S^*|}{k}\right)\tau} \notag\\
&\geq \lambda \left( \Ex{\mathcal{D}[N]}{\sum_{i \in S^*} \phi(i,S^*) \cdot r_i} - \tau\right) \notag,
\end{align*}
where the second inequality above follows since $\phi(i,\{i\}) \geq \phi(i,S^*)$ for every $i \in S^*$, by substitutability, and by the fact that $\frac{\phi(0,\{i\})}{\phi(i,\{i\})} = \frac{\phi(0,S^*)}{\phi(i,S^*)}$, by \Cref{cond:3}. The penultimate inequality follows from using (twice) that $(x)^+ \geq x$, for any real $x$, and reordering terms. The last inequality follows from $|S^*| \leq k$, by feasibility of the optimal solution, and from the identity $\phi(0,S^*) + \phi(S^*,S^*) = 1$. This concludes the proof. 
\end{proof}

Let us now bound the expectation of the second term in the RHS of decomposition \eqref{eq:decompositioncardinality}: 

\begin{lemma}\label{lem:cardinality:term2}
For any demand model satisfying (the weak version of) \Cref{cond:3}, it holds
$$
\Ex{\mathcal{D}[N]}{\phi(A_{\tau},A_{\tau}) + \sum_{i \in A_{\tau}}\frac{\phi(0,\{i\})}{k}\frac{\phi(i,A_{\tau})}{\phi(i,\{i\})}} \geq 1 - \lambda.
$$
\end{lemma}
\begin{proof}
By using (the weak version of) \Cref{cond:3}, we can see that
$$
\sum_{i \in A_{\tau}}\frac{\phi(0,\{i\})}{k}\frac{\phi(i,A_{\tau})}{\phi(i,\{i\})} \geq \sum_{i \in A_{\tau}}\frac{\phi(0,A_{\tau})}{k} =  \phi(0,A_{\tau})  \frac{|A_{\tau}|}{k} \geq \phi(0,A_{\tau}) \event{|A_{\tau}| = k}.
$$
By combining the above inequality with the definition of $\lambda$ the proof follows directly.
\end{proof}

We can now complete the proof of \Cref{thm:cardinality}.

\begin{proof}[Proof of \Cref{thm:cardinality}]
By taking expectation over equation \eqref{eq:decompositioncardinality} and using \Cref{lem:cardinality:term1,lem:cardinality:term2}, we get that 
\begin{align*}
\Ex{\mathcal{D}[N]}{f(A_\tau)} \geq \lambda \cdot \left(\Ex{\mathcal{D}[N]}{f(S^*)} - \tau\right) + (1-\lambda) \cdot \tau.
\end{align*}
The proof follows by setting $\tau = \frac{1}{2}\Ex{\mathcal{D}[N]}{f(S^*)}$, where $S^*$ is the optimal cardinality-constrained solution.    
\end{proof}

\section{Hardness Results}
\label{sec:lowerbounds}
In this section, we provide lower bounds on the achievable competitive guarantees for the sequential assortment selection problem under MNL demand; clearly, the following hardness results characterize the whole class of choice models satisfying \Cref{cond:1,cond:2} (and \Cref{cond:3}). The proofs of this section can be found in \Cref{app:hardness}. We start by proving that the original Prophet Inequality problem \cite{KS77,KS78} can be reduced to the sequential assortment selection problem (under MNL demand):

\begin{restatable}{theorem}{restatelowerboundobservation} 
\label{thm:reduction}
Any $\rho$-competitive policy for the unconstrained or cardinality-constrained version of sequential assortment selection (under MNL demand) can be transformed into a $\rho$-competitive policy for the original (choose-one) Prophet Inequality problem.
\end{restatable}
\begin{proof}[Proof sketch.]
Given an instance of the Prophet Inequality problem -- where the decision-maker chooses one among $n$ random rewards, $w_1, \dots, w_{n}$ -- we construct an ``equivalent'' instance of sequential assortment selection with the property that, whenever a subset $A$ of items is collected, then the total revenue (effectively) becomes $\min_{i \in A} w_i$. Under MNL demand, this can be achieved by setting "infinite" attractions to all items (excluding the outside option), in a way that the attraction of an item approaches infinity ``faster'' than that of any other item with higher revenue (reward). In this way, the collected item with smallest revenue (effectively) cannibalizes the demand of every other item in the assortment (including the outside option). The reduction then follows from simulating any (unconstrained or cardinality-constrained) $\rho$-competitive policy on this fictitious instance and mirroring its decisions up to the collection of the first item; after this point, the rest of the arriving items are rejected.
\end{proof}

By the above discussion it already becomes clear that both the unconstrained and the cardinality-constrained version of our problem inherit the hardness of the original Prophet Inequality setting \cite{HK92}:

\begin{restatable}{corollary}{restatelowerboundsimple} \label{thm:lowerboundsimple}
For any $\epsilon > 0$, there exists no policy with competitive guarantee smaller than $2-\epsilon$ for neither the unconstrained nor the cardinality/knapsack-constrained sequential assortment selection problem (under MNL demand).  
\end{restatable}

In the following result, we provide a refined lower bound for the unconstrained setting, by parameterizing all possible instances using as a metric the expected probability of purchasing in the optimal solution.

\begin{restatable}{theorem}{restatelowerboundstrong} \label{thm:lowerboundstrong}
Let $\gamma = \E{\psi(S^*)} \in [0,1]$ be the expected probability of purchasing (any item) in the optimal unconstrained assortment for some substitutable demand model satisfying \Cref{cond:1,cond:2} (that includes MNL). Then, for any $\epsilon > 0$, there exists no policy for the unconstrained sequential assortment selection problem under this model with competitive guarantee smaller than $1 + \gamma - \epsilon$.
\end{restatable}
\begin{proof}[Proof sketch.]
The lower bound is based on an instance with MNL demand and two arriving items $i_1$ and $i_2$. Fixing some $\delta \in (0,1)$ and $\kappa \in (0,1)$, the revenue and attraction of item $i_1$ become $r_1 = 1$ and $v_1 = \frac{1}{\delta}$. The attraction of the second item is $v_2 = 1$ and its revenue is either $r_2 = r = \frac{1}{\delta} \cdot \left(\frac{v_0}{v_2} + 1\right)$, with probability $\delta$, and $r_2 = 0$, otherwise. The attraction of the outside option satisfies $v_0 = \frac{1-\kappa}{\kappa} v_1$.
Let $S^*$ denote the optimal unconstrained solution and $A$ the assortment collected by the optimal policy for this instance. By taking the limit $\delta \to 0$, one can verify that $\E{f(S^*)} \to 1 + \kappa$, $\gamma \to \kappa$, and $\E{f(A)} \to 1$; then the statement follows by continuity.
\end{proof}

Notice that the above implies that \Cref{alg:unconstrained} yields the optimal guarantee for any value of $\gamma$. Interestingly, by observing the proof of \Cref{thm:lowerboundstrong}, one can see that the lower bound holds even for instances where either the revenues or the attractions of all items are deterministic.

\section{Extension to Knapsack Constraints with Random Sizes} \label{sec:knapsack}
In this section, we consider a more general and practical setting where any assortment is subject to knapsack constraints. Specifically, every item $i \in N$ is associated with a non-negative size $b_i$ and the total size of any feasible subset $S$ -- denoted by $b(S) = \sum_{i \in S} b_i$ -- cannot exceed some given (deterministic) budget $B$. In our sequential stochastic setting, we allow the size of each item $i$ to be random and arbitrarily correlated with its revenue and demand parameter, namely, the parameters $(r_i,d_i,b_i)$ are drawn from a joint distribution $\mathcal{D}_i$, independently of other items. 

For the above setting, by extending our techniques in \Cref{sec:cardinality}, we first provide a threshold-based policy for the practical setting where the maximum possible size of any item cannot exceed a specific fraction $\beta$ of the given budget, namely, $\max_{i \in {N} }b_i \leq \beta B$; without loss of generality we assume that $b_i \leq B$ for every $i \in N$ and, hence, $\beta \in [0,1]$. 
Our policy in this case can be thought of as a generalization of \Cref{alg:cardinality} for the case of cardinality constraints: every item $i \in N$ with $(r_i,d_i,b_i) \sim \mathcal{D}_i$ is accepted if $\phi(i,\{i\}) (r_i-\tau) \geq \frac{b_i \phi(0,\{i\})}{B}  \cdot \tau$ and accepting it would not violate the knapsack constraints. 

\begin{algorithm} \label{alg:budget}
Compute threshold $\tau = \frac{1}{2-\beta} \cdot \Ex{}{f(S^*)}$ and set ${A}_{\tau} \gets \emptyset$. For each arriving item $i \in N$, if $\phi(i,i) (r_i-\tau) \geq \frac{b_i  \phi(0,i)}{B}  \cdot \tau$ and $b(A_{\tau}) + b_i \leq B$, then add $i$ to $A_{\tau}$. 
\end{algorithm}

As we show in the next result, the competitive guarantee of \Cref{alg:budget} depends on the factor $\beta$. The proofs of this section can be found in \Cref{app:omitted}.

\begin{restatable}{theorem}{restatebudget} \label{thm:budget}
For the case of MNL demand (or any substitutable choice model satisfying \Cref{cond:1,cond:2,cond:3}) and any instance where $\max_{i \in N} \frac{b_i}{B} \leq \beta$ almost surely, \Cref{alg:budget} is $\frac{2-\beta}{1-\beta}$-competitive for the knapsack-constrained setting of sequential assortment selection. Moreover, if a demand model satisfies only the weak version of \Cref{cond:3} (together with the rest of our conditions), then by setting $\tau = \frac{1}{3-\beta} \E{f(S^*)}$ \Cref{alg:budget} becomes $\frac{3-\beta}{1-\beta}$-competitive.
\end{restatable}

We observe that the guarantees of \Cref{alg:budget} become meaningless when $\beta$ approaches $1$; to get a constant-competitive policy (for any $\beta$), we leverage the subadditivity of the assortment objective under substitutable demands and combine it with a fairly standard trick. Specifically, consider partitioning any realized instance into two disjoint subsets $Q = \{i \in {N} \mid b_i \leq \frac{B}{2}\}$ and $V = \{i \in N \mid b_i > \frac{B}{2}\}$, and let us refer to them as the ``small'' and ``large'' items, respectively. Let us denote by $g(S) = \max_{T \subseteq S, b(T) \leq B} f(T) $ the optimal feasible assortment contained in a subset $S$ of the items for a given realization.
Using the fact that the function $f(\cdot)$ is subadditive for any substitutable choice model, for any feasible assortment $A$, we have 
\begin{align} \label{eq:budget:combine}
f(A) \leq f(A\cap V) + f(A \cap Q) \leq g(V) + g(Q) \leq 2\max(g(V), g(Q)).  
\end{align}
The above inequality suggests that either the maximum between the optimal total revenue restricted to the items in $Q$ and that in $V$ is a $2$-approximation to the optimal total revenue (over all $N$). Observe that, in our sequential setting, one compete against $\E{g(Q)}$ by applying \Cref{alg:budget} for $\beta = \frac{1}{2}$ only on the items in $Q$. Similarly, since $b_i > \frac{B}{2}$ for any $i \in V$, any feasible solution contained in $V$ must be a singleton, since, otherwise, the budget constraint would be violated. In this case, using \Cref{alg:cardinality} for $k=1$ on the items in $V$, gives a $2$-competitive policy against $\E{g(V)}$. 

The following randomized policy naturally interpolates between the above two cases: 

\begin{algorithm} \label{alg:budgetfull}
Let $g(Q)$ and $g(V)$ denote the total revenue of the maximum feasible assortment restricted to items with $b_i \leq \frac{B}{2}$ and $b_i > \frac{B}{2}$, respectively. Flip a (biased) coin $C \in \{{H}, {T}\}$ with probability $\Pro{C = {H}} = \frac{3}{5}$. If $C = H$, then run \Cref{alg:budget} with threshold $\tau = \frac{2}{3} \E{g(Q)}$, automatically rejecting any $i \in N$ with $b_i > \frac{B}{2}$. In case $C = T$, run \Cref{alg:cardinality} with $k = 1$ and $\tau = \frac{1}{2} \E{g(V)}$, automatically rejecting any item with $b_i \leq \frac{B}{2}$. 
\end{algorithm}

In the next result, we provide a bound to the competitive guarantee of \Cref{alg:budgetfull}: 

\begin{restatable}{theorem}{restatebudgetfull} \label{thm:budgetfull}
For the case of MNL demand (or any substitutable choice model satisfying \Cref{cond:1,cond:2,cond:3}), \Cref{alg:budgetfull} is $5$-competitive for the sequentual assortment selection problem with knapsack constraints. Moreover, an adaptation of \Cref{alg:budgetfull} becomes $8$-competitive for substitutable models satisfying only the weak version of \Cref{cond:3}, together with \Cref{cond:1,cond:2}.
\end{restatable}

\section{Applications and Implementation} \label{sec:applications}

In this section, we present practical examples of substitutable demand models -- other than MNL -- which satisfy Conditions \ref{cond:1}, \ref{cond:2}, and the weak version of \Cref{cond:3}; for all these models, our results imply a $(1+\gamma)$-competitive policy for the unconstrained setting of sequential assortment selection, where $\gamma = \E{\phi(S^*,S^*)}$. In addition, our results yield $3-$ and $8$-competitive policies cardinality and knapsack constraints, respectively. The proofs of this section can be found in \Cref{app:demand}. In \Cref{app:implementation}, we discuss issues related to the efficient implementation of our policies in terms of computational and sample complexity.

\paragraph{Independent random utility model with deterministic outside.} In RUM each item $i \in N$ is associated with a (continuous) utility distribution $\mathcal{U}_i$, with $ \mathcal{U}_0$ being that of the outside option. When faced with an assortment $S\subseteq N$, a random utility $u_{i} \sim \mathcal{U}_i$ is realized for each option $i \in N \cup \{0\}$, and the buyer chooses the one with the highest realization (including potentially the outside). The choice probability for any assortment $S$ and item $i \in S$ can be thus described as
\begin{align*}
\phi(i,S) = \Pro{u_i \geq u_j, \forall j \in S \cup \{0\} \setminus \{i\}}.
\end{align*}
We consider a special class of RUM, where the utilities of all items are independently distributed, and the utility of the outside option is deterministic. A simple way to cast the above model as an instance of sequential assortment selection is the following: every item $i\in N$ is associated with a family of $m_i$ utility distributions with known description, denoted by $\{\mathcal{U}^{(l)}_i\}_{l \in [m_i]}$. The demand parameter of each item $i \in N$ encodes the index of the distribution used for realizing $u_i$; that is, an independently realized pair $(r_i,d_i) \sim \mathcal{D}_i$ with $d_i \in [m_i]$, suggests that $i$ has revenue $r_i$ and utility $u_i \sim \mathcal{U}^{d_i}_i$. 

\begin{restatable}{proposition}{restateRUM} 
The independent RUM with deterministic outside utility satisfies Conditions \ref{cond:1}, \ref{cond:2}, and the weak version of \Cref{cond:3}.
\end{restatable}

\paragraph{Generalized attraction model with deterministic shadow attractions.} The MNL model fails to capture the negative effect (for example, due to competition) of not including an item in an assortment. To address this issue, Gallego et al. \cite{gallego2015general} introduce the Generalized Attraction Model (GAM), which strictly generalizes MNL through the notion of shadow attractions: each item $i \in N$ is associated with a non-negative revenue $r_i$, an attraction $v_i$, and a shadow attraction $w_i \in [0, v_i]$; note that \cite{gallego2015general} require the shadow attraction of each item to be at most its (actual) attraction, in order to make their model compatible with Luce's axioms \cite{luce1959individual}. Under the GAM model, for any assortment $S$ and item $i \in S$, we have 
$$
\phi(i,S) = \frac{v_i}{v_0 + v(S) + w(N \setminus S)},
$$
where $v_0$ is the deterministic attraction of the outside option. The choice probability of the outside option for assortment $S$ is naturally defined as $\phi(0,S) = \frac{v_0 + w(N \setminus S)}{v_0 + v(S) + w(N \setminus S)}$.

Adapting the above model to our sequential and stochastic setting is rather straightforward: for each item $i \in N$, there is a known joint distribution $(r_i,d_i) \sim \mathcal{D}_i$ (independent for each item), where the demand parameter $d_i$ encodes the value of $v_i$. We remark that we make the additional assumption of the shadow attractions being deterministic and known to the decision-maker; this seems to be essential since, otherwise, the probability $\phi(i,\{i\})$ could not be computed by simply observing the parameters of item $i$ (and, thus, \Cref{cond:1} cannot hold). To ensure compatibility with GAN, we also assume that the value $v_i$ is always greater than $w_i$, for each item $i \in N$ (yet it can still be random and arbitrarily correlated with $r_i$).


\begin{restatable}{proposition}{restateGAM} 
The GAM with deterministic shadow attractions satisfies Conditions \ref{cond:1}, \ref{cond:2}, and the weak version of \Cref{cond:3}.
\end{restatable}

\paragraph{Lowest considered fare model.} 
In the lowest considered fare (LCF) model, each item $i \in N$ is associated with a revenue $r_i$, a fare $f_i \in [0,\infty)$, and a consideration parameter $q_i \in [0,1]$. The customer's choice is made as follows: given an assortment $S$, a consideration set $Q \subseteq S$ is constructed, by including each item $i \in S$ to $Q$, independently, with probability $q_i$. Then, the buyer purchases the item in $Q$ with minimum fare, that is $\arg\min_{i \in Q} f_i$; if the consideration set is empty, then the buyer purchases no item. In the case where more than one minimizers exist, we break ties consistently using a lexicographic order. In this context, for every assortment $S$ and item $i \in S$, we have
$$
\phi(i,S) = q_i \cdot \prod_{\substack{j \in S \\ f_j < f_i}} \left(1 - q_j\right).
$$
The above is an variation of the lowest open fare (LOF) model, considered by Talluri and van Ryzin \cite{talluri2004revenue}, where the buyer observes only the item $\arg\min_{i \in S} r_i$ and purchases it with probability $q_i$ (the fare is implicitly proportional to the revenue). The latter is a simpler model and its natural sequential variant can be easily reduced to the Prophet Inequality problem. In fact, it can be easily verified that the optimal assortment in LOF contains without loss of generality only a single item, while in LCF contains all items with revenues above some threshold.

One can construct a natural sequential version of the above model by encoding in the demand parameter of each item $i \in N$ its fare and consideration probability, namely, $d_i \equiv (f_i,q_i)$ for each $i \in N$. Then, for each item $i \in N$ the pair $(r_i,d_i)$ is drawn from a known joint distribution $\mathcal{D}_i$, independently of other items. 

\begin{restatable}{proposition}{restateLCF} 
The lowest considered fare (LCF) model satisfies Conditions \ref{cond:1}, \ref{cond:2}, and the weak version of \Cref{cond:3}.
\end{restatable}

\section*{Conclusion and Further Directions}
Motivated by the practical issue of assortment optimization under uncertainty, we formulated the problem of sequential assortment selection. We proposed simple optimal and near-optimal policies for a class of substitutable demand models (including the celebrated MNL demand), and several common combinatorial constraints. From a theoretical viewpoint, our results demonstrate another interesting domain -- beyond the original (choose-one) Prophet Inequality setting -- where single-threshold rules can be worst-case optimal.


Our work leaves behind numerous directions for future research. Starting from the most immediate, we believe that the guarantees we provide for knapsack constraints and MNL demand can be improved, yet it seems unlikely that this can be achieved via a fixed threshold (see \cite{jiang2022tight,dutting2020prophet}). Related to that, it would be interesting to explore different families of constraints for which the offline assortment selection problem is known to be solvable. As an example, it would be interesting to explore whether the (adaptive-threshold) techniques of Kleinberg and Weinberg \cite{KW12} or Dütting et al. \cite{dutting2020prophet} can be extended to provide (optimal) competitive guarantees in the case of matroid constraints. A different direction would be to study the problem under other common demand models (e.g., the nested logit \cite{davis2014assortment,gallego2014constrained}) which fail to satisfy our conditions. Finally, one could study the problem in the regime of unknown and arbitrarily chosen item parameters but uniformly random (secretary) arrivals \cite{dynkin1963optimal} (or combinations \cite{esfandiari2017prophet}).

\bibliographystyle{alpha}
\bibliography{ref}

\newpage 
\appendix

\section{Appendix: Omitted Proofs} \label{app:omitted}

\restatecardinalityweak*

\begin{proof}
The proof is essentially identical with that of \Cref{thm:cardinality} and relies on proving an alternative (weaker) version of \Cref{lem:cardinality:term1}, for the case where the weak version of \Cref{cond:3} is satisfied. Recall the decomposition \eqref{eq:decompositioncardinality} of the total revenue collected by the policy:
\begin{align*}
f(A_{\tau}) = \sum_{i \in A_{\tau}} \frac{\phi(i,A_{\tau})}{\phi(i,\{i\})} \left( \phi(i,\{i\})(r_i - \tau)^+ -\frac{\phi(0,\{i\})}{k} \tau\right)^+ + \left(\phi(A_{\tau},A_{\tau}) + \sum_{i \in A_{\tau}}\frac{\phi(0,\{i\})}{k}\frac{\phi(i,A_{\tau})}{\phi(i,\{i\})} \right) \tau,
\end{align*}
and let us define constants $\lambda = \Ex{\mathcal{D}[N]}{\phi(0,A_{\tau}) \event{|A_{\tau}| < k}}$ and $1-\lambda = \Ex{\mathcal{D}[N]}{\phi(0,A_{\tau}) \event{|A_{\tau}| = k} + \phi(A_{\tau},A_{\tau})}$, exactly as in the proof of \Cref{thm:cardinality}. 

We observe that \Cref{lem:cardinality:term2} only makes use of the weak version of \Cref{cond:3}, and hence, the expectation of the right term in the above decomposition is still lower bounded by $(1- \lambda)\tau$. For the left term, by following the proof of \Cref{lem:cardinality:term1} and without using \Cref{cond:3}, we can derive that 
\begin{align*} 
\Ex{\mathcal{D}[N]}{\sum_{i \in A_{\tau}} \frac{\phi(i,A_{\tau})}{\phi(i,\{i\})} \left( \phi(i,\{i\})(r_i - \tau)^+ -\frac{\phi(0,\{i\})}{k} \tau\right)^+} &\geq \lambda \Ex{\mathcal{D}[N]}{\sum_{i \in S^*} \phi(i,\{i\})\left((r_i - \tau)^+ -\frac{\phi(0,\{i\})}{\phi(i,\{i\})} \frac{\tau}{k} \right)^+} \\
&\geq \lambda \Ex{\mathcal{D}[N]}{\sum_{i \in S^*} \phi(i,\{i\})\left((r_i - \tau)^+ -\frac{\tau}{\phi(i,\{i\}) \cdot k} \right)} \\
&= \lambda \Ex{\mathcal{D}[N]}{\sum_{i \in S^*} \phi(i,\{i\})(r_i - \tau)^+ -\frac{|S^*|}{k} \tau} \\
&\geq \lambda \Ex{\mathcal{D}[N]}{\sum_{i \in S^*} \phi(i,S^*)(r_i - \tau)^+ -\frac{|S^*|}{k} \tau} \\
&\geq \lambda \left(\Ex{\mathcal{D}[N]}{f(S^*)} - \Ex{\mathcal{D}[N]}{\phi(S^*,S^*) +\frac{|S^*|}{k}} \tau\right),
\end{align*}
where in the second inequality above, we simply upper-bound $\phi(0,\{i\}) \leq 1$.

By combining the above, for the expectation of $f(A_{\tau})$, we can see that 
\begin{align*}
\Ex{\mathcal{D}[N]}{f(A_{\tau})} &\geq \lambda \left(\Ex{\mathcal{D}[N]}{f(S^*)} - \Ex{\mathcal{D}[N]}{\phi(S^*,S^*) +\frac{|S^*|}{k}} \tau\right) + (1- \lambda)\tau \\
& \geq \lambda \left(\Ex{\mathcal{D}[N]}{f(S^*)} - \left(\Ex{\mathcal{D}[N]}{\phi(S^*,S^*)} +1 \right) \tau\right) + (1- \lambda)\tau,
\end{align*}
using the fact that $|S^*| \leq 1$.

Let us denote by $\gamma = \Ex{\mathcal{D}[N]}{\phi(S^*,S^*)}$ the expected probability of purchasing (any item) in the optimal cardinality-constrained assortment. By choosing threshold $\tau = \frac{1}{2+\gamma} \cdot \Ex{\mathcal{D}[N]}{f(S^*)}$, the policy becomes $(2+\gamma)$-competitive. Finally, by setting threshold $\tau = \frac{1}{3} \Ex{\mathcal{D}[N]}{f(S^*)}$ and noticing that $\gamma \leq 1$, the policy becomes $3$-competitive.
\end{proof}

\restatebudget*

\begin{proof}
We remark that several parts of this proof are written succinctly, since they follow closely the corresponding parts for the cardinality-constrained setting; hence, we advice the reader to first understand the proof in \Cref{sec:cardinality:proof}.

Before we proceed to the main proof, let us extend the notation appearing in \Cref{sec:cardinality} for the (special) case of cardinality constraints. Let us fix a threshold $\tau > 0$ and denote by $E_{\tau} = \{i \in {N} \mid \phi(i,\{i\})(r_i-\tau) \geq \frac{b_i  \phi(0,\{i\})}{B} \tau\}$ the (random) subset of items that satisfy the threshold condition of \Cref{alg:budget}. Similarly to the case of cardinality constraints, a worst-case arrival order $\sigma$ is computed by an adversary after observing the realization of the parameters of all items. As before, for the particular case of \Cref{alg:budget} we assume without loss of generality that there exists a unique mapping from realizations to worst-case orderings (by breaking ties-lexicographically). Let $A^{\prec i}_\tau$ be the subset of items collected by our policy exactly before the arrival of item $i$ (under the worst-case arrival order $\sigma$) and let $A_{\tau} \subseteq E_{\tau}$ be the collected assortment after the arrival of all items. Finally, we remark that, since the cases where an item has size $b_i > \beta \cdot B$ have zero measure, we simply assume that $b_i \leq \beta \cdot B$ (deterministically) when arguing inside expectations.

By working along the lines of the proof of \Cref{thm:cardinality}, the total revenue collected by \Cref{alg:budget} using a fixed threshold $\tau$ can be rewritten as
\begin{align}
f(A_{\tau}) &= \sum_{i \in A_{\tau}} \phi(i,A_{\tau}) \cdot r_i \notag \\
&= \sum_{i \in A_{\tau}} \frac{\phi(i,A_{\tau})}{\phi(i,\{i\})} \left( \phi(i,\{i\})(r_i - \tau)^+ -\frac{b_i \phi(0,\{i\})}{B} \tau\right)^+ + \left(\phi(A_{\tau},A_{\tau}) + \sum_{i \in A_{\tau}}\frac{b_i\phi(0,\{i\})}{B}\frac{\phi(i,A_{\tau})}{\phi(i,\{i\})} \right) \tau, \label{eq:decompositionbudget}
\end{align}
where we observe that any item $i \in A_{\tau}$ must satisfy $r_i \geq \tau $, by construction of the policy.

Let us define constant $\lambda = \Ex{\mathcal{D}[{N}]}{\phi(0,A_{\tau})\event{\frac{b(A_{\tau})}{B} < 1-\beta}}$  and, using the uniqueness of the worst-case arrival for each realization, observe that
$1-\lambda = \Ex{\mathcal{D}[{N}]}{\phi(0,A_{\tau})\event{\frac{b(A_{\tau})}{B} \geq 1-\beta} + \phi(A_{\tau},A_{\tau})}.$

By working similarly to the cardinality-constrained case, we first bound the expectation of the first term in decomposition \eqref{eq:decompositionbudget}. 

\begin{lemma}\label{lem:budget:term1}
For any substitutable demand model satisfying \Cref{cond:1,cond:2,cond:3}, it holds
$$
\Ex{\mathcal{D}[N]}{\sum_{i \in A_{\tau}} \frac{\phi(i,A_{\tau})}{\phi(i,\{i\})} \left( \phi(i,\{i\})(r_i - \tau)^+ -\frac{b_i \phi(0,\{i\})}{B} \tau\right)^+} \geq \lambda \cdot \left(\Ex{\mathcal{D}[N]}{f(S^*)} - \tau \right),
$$
where $S^*$ denotes the optimal knapsack-constrained solution for a given realization. Moreover, for  any substitutable model satisfying Conditions \ref{cond:1}, \ref{cond:2} and the weak version of \Cref{cond:3}, we have
$$
\Ex{\mathcal{D}[N]}{f(A_{\tau})} \geq \lambda \left(\Ex{\mathcal{D}[N]}{f(S^*)} - \Ex{\mathcal{D}[N]}{\phi(S^*,S^*) +\frac{b(S^*)}{B}} \tau\right) + (1- \lambda)\tau.
$$
\end{lemma}
\begin{proof}
For any $i \in N$, let us denote $\zeta(i) = \left( \phi(i,\{i\})(r_i - \tau)^+ -\frac{b_i\phi(0,\{i\})}{B} \cdot \tau\right)^+$ and observe that for any $i \in N$ that does not satisfy the threshold condition of \Cref{alg:budget}, it must hold that $\zeta(i) = 0$. Notice, further, that whenever $\zeta(i) > 0$ for some $i \in N$, it must be that $\phi(i,\{i\})>0$. By working along the lines of \Cref{lem:cardinality:term1}, we can see that 
\begin{align*}
\Ex{\mathcal{D}[N]}{\sum_{i \in A_{\tau}} \frac{\phi(i,A_{\tau})}{\phi(i,\{i\})} \left( \phi(i,\{i\})(r_i - \tau)^+ -\frac{b_i \phi(0,\{i\})}{B} \tau\right)^+} &= \Ex{\mathcal{D}[N]}{\sum_{i \in N} \frac{\phi(i,A_\tau)}{\phi(i,\{i\})} \zeta(i) \event{b(A^{\prec i}_{\tau}) + b_i \leq B}} \\
&\geq \Ex{\mathcal{D}[N]}{\sum_{i \in N} \phi\left(0,E_{\tau}\setminus \{i\}\right)  \zeta(i) \event{\frac{b(E_{\tau} \setminus \{i\})}{B} \leq 1- \beta}},
\end{align*}
where the inequality follows since $\phi(i,A_{\tau}) \geq \phi(i,\{i\}) \phi(0,E_{\tau} \setminus \{i\})$, by \Cref{cond:2} and substitutability using that $A_{\tau} \subseteq E_{\tau}$, combined with the fact that 
$$
\big\{ b(E_{\tau} \setminus \{i\}) \leq (1- \beta) B \big\} \subseteq \big\{ b(A_{\tau}^{\prec i}) \leq (1- \beta) B \big\} \subseteq \big\{ b(A_{\tau}^{\prec i}) + b_i \leq B \big\},
$$
which holds by definition of $E_{\tau}$ and since $b_i \leq \beta B$, almost surely (recall, the case $b_i > \beta B$ has zero measure).

We notice that $\zeta(i)$ and $b(E_{\tau} \setminus \{i\})$ only depend on the randomness of $\mathcal{D}_i$ and $\mathcal{D}[N \setminus \{i\}]$, respectively, by \Cref{cond:1}. Thus, working similarly to the proof of \Cref{lem:cardinality:term1}, we get
\begin{align*}
\Ex{\mathcal{D}[N]}{\sum_{i \in N} \phi\left(0,E_{\tau}\setminus \{i\}\right)  \zeta(i) \event{\frac{b(E_{\tau} \setminus \{i\})}{B} \leq 1- \beta}} &\geq \sum_{i \in N} \Ex{\mathcal{D}_i}{\zeta(i)}\Ex{\mathcal{D}[N]}{\phi\left(0,E_{\tau}\right)   \event{\frac{b(E_{\tau})}{B} \leq 1- \beta}}\\
&= \sum_{i \in N} \Ex{\mathcal{D}_i}{\zeta(i)}\Ex{\mathcal{D}[N]}{\phi\left(0,A_{\tau}\right)   \event{\frac{b(A_{\tau})}{B} \leq 1- \beta}}\\
&\geq \lambda \cdot \Ex{\mathcal{D}}{\sum_{i \in N} \left( \phi(i,\{i\})(r_i - \tau)^+ -\frac{b_i\phi(0,\{i\})}{B}  \tau\right)^+},
\end{align*}
where in the first inequality above we use substitutability and the fact that $b(\cdot)$ is monotone non-decreasing. In the equality we use the fact that whenever $b(E_{\tau}) \leq (1- \beta)B$, then the sets $E_{\tau}$ and $A_{\tau}$ must coincide, while in the second inequality we use the definition of $\lambda$ together with the fact that $\zeta(i)$ only depends on $\mathcal{D}_i$ for any item $i \in N$.

Let $S^*$ be the optimal knapsack-constrained solution. Notice that $b(S^*) \leq B$, by feasibility constraints, and that for any item $i \in S^*$ it holds $\phi(i,i) > 0$, without loss of generality. By working similarly to the proof of \Cref{lem:cardinality:term1}, we can conclude that
\begin{align*} 
\Ex{\mathcal{D}[N]}{\sum_{i \in A_{\tau}} \frac{\phi(i,A_{\tau})}{\phi(i,\{i\})} \left( \phi(i,\{i\})(r_i - \tau) -\frac{b_i\phi(0,\{i\})}{B} \tau\right)^+} &\geq \lambda \Ex{\mathcal{D}[N]}{\sum_{i \in S^*} \phi(i,\{i\})\left( (r_i - \tau)^+ -\frac{b_i}{B} \frac{\phi(0,\{i\})}{\phi(i,\{i\})} \tau\right)^+} \\
&\geq \lambda \Ex{\mathcal{D}[N]}{\sum_{i \in S^*} \phi(i,S^*)\left( (r_i - \tau)^+ -\frac{b_i}{B} \frac{\phi(0,S^*)}{\phi(i,S^*)} \tau\right)^+} \\
&\geq \lambda \cdot \left(\Ex{\mathcal{D}[N]}{\sum_{i \in S^*} \phi(i,S^*) \cdot r_i} - \tau \right),
\end{align*}
where the second inequality above follows from substitutability, which implies $\phi(i,i) \geq \phi(i,S^*)$ for every $i \in S^*$, and by \Cref{cond:3} which suggests that $\frac{\phi(0,\{i\})}{\phi(i,\{i\})} = \frac{\phi(0,S^*)}{\phi(i,S^*)}$. The last inequality follows from the facts that $(x)^+ \geq x$ and $b(S^*) \leq B$. That proves the first part of the statement. 

Alternatively, one could use the weak version of \Cref{cond:3} and provide the following lower bound:

\begin{align*} 
\Ex{\mathcal{D}[N]}{\sum_{i \in A_{\tau}} \frac{\phi(i,A_{\tau})}{\phi(i,\{i\})} \left( \phi(i,\{i\})(r_i - \tau) -\frac{b_i\phi(0,\{i\})}{B} \tau\right)^+} 
&\geq \lambda \Ex{\mathcal{D}[N]}{\sum_{i \in S^*} \phi(i,\{i\})\left( (r_i - \tau)^+ -\frac{b_i}{B} \frac{\phi(0,\{i\})}{\phi(i,\{i\})} \tau\right)^+} \\
&\geq \lambda \Ex{\mathcal{D}[N]}{\sum_{i \in S^*} \phi(i,\{i\})\left( (r_i - \tau)^+ -\frac{b_i}{B \cdot \phi(i,\{i\})} \tau\right)} \\ 
&= \lambda  \Ex{\mathcal{D}[N]}{\sum_{i \in S^*} \phi(i,\{i\})(r_i - \tau)^+ -\frac{b(S^*)}{B} \tau)} \\ 
&\geq \lambda  \Ex{\mathcal{D}[N]}{\sum_{i \in S^*} \phi(i,S^*)(r_i - \tau)^+ -\frac{b(S^*)}{B} \tau)} \\ 
&\geq \lambda \left(\Ex{\mathcal{D}[N]}{\sum_{i \in S^*} \phi(i,S^*) r_i} - \Ex{\mathcal{D}[N]}{\phi(S^*,S^*) +\frac{b(S^*)}{B}} \tau \right),
\end{align*}
where in the second inequality above we use that $\phi(0,\{i\}) \leq 1$ and $(x)^+ \geq x$. The third inequality follows from substitutability and the last from reordering terms and using again $(x)^+ \geq x$. This concludes the proof.
\end{proof}

We can bound the expectation of the second term in decomposition \eqref{eq:decompositionbudget} using the following lemma:

\begin{lemma}\label{lem:budget:term2}
For any demand model satisfying (the weak version of) \Cref{cond:3}, it holds
$$
\Ex{\mathcal{D}[N]}{\phi(A_{\tau},A_{\tau}) + \sum_{i \in A_{\tau}}\frac{b_i\phi(0,\{i\})}{B}\frac{\phi(i,A_{\tau})}{\phi(i,\{i\})} } \geq (1-\beta)(1 - \lambda),
$$
where $\beta \in [0,1]$ is such that $b_i \leq \beta B$ for every item $i \in N$, almost surely.
\end{lemma}
\begin{proof}
By using (the weak version of) \Cref{cond:3}, we can see that
$$
\sum_{i \in A_{\tau}}\frac{b_i\phi(0,\{i\})}{B}\frac{\phi(i,A_{\tau})}{\phi(i,\{i\})} \geq \phi(0,A_{\tau}) \cdot \frac{b(A_{\tau})}{B} \geq (1-\beta) \cdot \phi(0,A_{\tau}) \cdot \event{\frac{b(A_{\tau})}{B} \geq 1-\beta} .
$$
Therefore, using the definition of $\lambda$, we can conclude that
\begin{align*}
\Ex{\mathcal{D}[N]}{\phi(A_{\tau},A_{\tau}) + \sum_{i \in A_{\tau}}\frac{b_i\phi(0,\{i\})}{B}\frac{\phi(i,A_{\tau})}{\phi(i,\{i\})} } &\geq \Ex{\mathcal{D}[N]}{\phi(A_{\tau},A_{\tau}) + (1-\beta) \cdot \phi(0,A_{\tau}) \cdot \event{\frac{b(A_{\tau})}{B}>1-\beta} }  \\
&\geq (1-\beta) \cdot \Ex{\mathcal{D}[N]}{\phi(A_{\tau},A_{\tau}) +  \phi(0,A_{\tau}) \cdot \event{\frac{b(A_{\tau})}{B}>1-\beta} }  \\
&\geq (1-\beta) (1 - \lambda).
\end{align*}
\end{proof}

By taking expectation over decomposition \eqref{eq:decompositionbudget} and using \Cref{lem:budget:term1,lem:budget:term2}, for any substitutable choice model satisfying \Cref{cond:1,cond:2,cond:3}, we get 
\begin{align*}
\Ex{\mathcal{D}[N]}{f(A_\tau)} \geq \lambda \cdot \left(\Ex{\mathcal{D}[N]}{f(S^*)} - \tau\right) + (1-\beta)(1-\lambda) \cdot \tau.
\end{align*}
Therefore, by setting $\tau = \frac{1}{2-\beta}\Ex{\mathcal{D}[N]}{f(S^*)}$, where $S^*$ is the optimal knapsack-constrained solution, \Cref{alg:budget} becomes an $\frac{2-\beta}{1-\beta}$-competitive policy for this class. 

Finally, for  any substitutable model satisfying Conditions \ref{cond:1}, \ref{cond:2} and the weak version of \Cref{cond:3}, again by applying the above lemmas, we have 
\begin{align*}
\Ex{\mathcal{D}[N]}{f(A_\tau)} \geq \lambda \cdot \left(\Ex{\mathcal{D}[N]}{f(S^*)} - \Ex{\mathcal{D}[N]}{\phi(S^*,S^*) +\frac{b(S^*)}{B}} \tau\right) + (1-\beta)(1-\lambda) \cdot \tau.
\end{align*}
In this case, noticing that $\Ex{\mathcal{D}[N]}{\phi(S^*,S^*) +\frac{b(S^*)}{B}} \leq 2$ and setting $\tau = \frac{1}{3-\beta}\Ex{\mathcal{D}[N]}{f(S^*)}$, \Cref{alg:budget} becomes an $\frac{3-\beta}{1-\beta}$-competitive policy.
\end{proof}

\restatebudgetfull*
\begin{proof}
We prove the result for any substitutable demand model satisfying \Cref{cond:1,cond:2,cond:3}; then we show how \Cref{alg:budgetfull} can be adapted to become $8$-competitive for the case where only the weak version of \Cref{cond:3} is satisfied.

Let us denote by $A \subseteq {N}$ the subset of items collected by \Cref{alg:budgetfull}, and by $A_Q$ and $A_V$ the items that the policy collects in the case where $C=H$ and $C=T$, respectively. We remark that since the policy collects either only items in $Q$ or only in $V$, then we can still assume that the worst-case arrival order is without loss of generality unique for every realization; this, for example, can be achieved by combining the worst-case arrival orders for each case, breaking ties lexicographically. 

Since the random decision of which case to follow is made by the policy independently of any realizations, the expected total revenue of \Cref{alg:budgetfull} becomes 
$$
\Ex{\mathcal{D}}{f(A)} = \frac{3}{5} \cdot \Ex{\mathcal{D}}{f(A_Q)} + \frac{2}{5} \cdot \Ex{\mathcal{D}}{f(A_V)}.
$$

In the case where only items from $Q$ are collected, by applying \Cref{thm:budget} with $\beta = \frac{1}{2}$, we get that $\Ex{\mathcal{D}}{f(A_Q)} \geq \frac{1}{3} \cdot \Ex{\mathcal{D}}{g(Q)}$. We remark that the fact that items in $V$ are automatically skipped by the policy does not affect the competitive guarantee of \Cref{thm:budget}, relative to $\Ex{\mathcal{D}}{g(Q)}$. An easy way to see this is to consider the following transformation of an instance: the parameters of every item $i \in {N}$ are generated by first drawing a triple $({r}_i, {d}_i, {b}_i) \sim {\mathcal{D}}_i$ and then return the triple $(\Tilde{r}_i, \Tilde{d}_i, \Tilde{b}_i)$ such that $\Tilde{d}_i = d_i$, $\Tilde{b}_i = b_i$, and $\Tilde{r}_i = r_i \cdot \event{b_i \leq \frac{B}{2}}$. In other words, the revenue of an item is (fictitiously) set to zero if its size is greater than $\frac{B}{2}$. Let $\Tilde{\mathcal{D}}_i$ denote the transformed parameter distribution for each item $i \in N$; note that $\Tilde{\mathcal{D}}$ is a valid input for \Cref{alg:budget}, since $\Tilde{b}_i$ is allowed to be correlated with $\Tilde{r}_i$ and $\Tilde{d}_i$. In this case, the optimal total revenue is without loss of generality supported only on the items in $Q$ and, hence, $\Ex{\Tilde{\mathcal{D}}}{f(S^*)} = \Ex{{\mathcal{D}}}{g(Q)}$. 

By working as above for the case where only items in $V$ are collected, \Cref{thm:cardinality} with $k=1$ implies that $\Ex{\mathcal{D}}{f(A_V)} \geq \frac{1}{2} \cdot \Ex{{\mathcal{D}}}{g(V)}$. Using the fact that the (random) sets $Q$ and $V$ are a partition of ${N}$, we get
\begin{align*}
    \Ex{\mathcal{D}}{f(A)} = \frac{3}{5} \cdot \Ex{\mathcal{D}}{f(A_Q)} + \frac{2}{5} \cdot \Ex{\mathcal{D}}{f(A_V)} \geq \frac{1}{5} \cdot \Ex{\mathcal{D}}{g(Q) + g(V)} \geq \frac{1}{5} \cdot \Ex{\mathcal{D}}{f(S^*)},
\end{align*}
where the last inequality follows from equation \eqref{eq:budget:combine}. This concludes the proof for the case of substitutable demand models satisfying \Cref{cond:1,cond:2,cond:3}.

The following policy applies to substitutable models satisfying only the weak version of \Cref{cond:3}:

\begin{algorithm} 
Let $g(Q)$ and $g(V)$ denote the total revenue of the maximum feasible assortment restricted to items with $b_i \leq \frac{B}{2}$ and $b_i > \frac{B}{2}$, respectively. Flip a (biased) coin $C \in \{{H}, {T}\}$ with probability $\Pro{C = {H}} = \frac{5}{8}$. If $C = H$, then run \Cref{alg:budget} with threshold $\tau = \frac{2}{5} \E{g(Q)}$, automatically rejecting any $i \in N$ with $b_i > \frac{B}{2}$. In case $C = T$, run \Cref{alg:cardinality} with $k = 1$ and $\tau = \frac{1}{3} \E{g(V)}$, automatically rejecting any item with $b_i \leq \frac{B}{2}$. 
\end{algorithm}
In this case, by working exactly as above and using \Cref{thm:cardinality3,thm:budget}, we can derive that 
\begin{align*}
    \Ex{\mathcal{D}}{f(A)} = \frac{5}{8} \cdot \Ex{\mathcal{D}}{f(A_Q)} + \frac{3}{8} \cdot \Ex{\mathcal{D}}{f(A_V)} \geq \frac{1}{8} \cdot \Ex{\mathcal{D}}{g(Q) + g(V)} \geq \frac{1}{8} \cdot \Ex{\mathcal{D}}{f(S^*)},
\end{align*}
thus concluding the proof. 
\end{proof}

\section{Appendix: Hardness Results} \label{app:hardness}

\restatelowerboundobservation*

\begin{proof}
Let $A$ be a given (unconstrained or cardinality-constrained) $\rho$-competitive policy for sequential assortment selection. It suffices to show that, for any $\epsilon > 0$, one can construct a policy for the original Prophet Inequality with $n$ random rewards, $w_1, \dots, w_{n}$, which collects -- in expectation -- at least $\frac{1}{\rho} \cdot \E{\max_{i \in N} w_i} - \epsilon$. As we show in this proof, this can be achieved for the case of MNL demand.

For any instance of the original Prophet Inequality, we can construct an instance of the unconstrained sequential assortment selection problem under MNL demand as follows: for each $i \in [n]$, we consider an item $i$ such that $r_i \sim w_i$, namely, its revenue is identically distributed with reward $w_i$. Without loss of generality, let us assume that the distribution of the latter is strictly positive and continuous. The attraction of the outside option is set to zero, that is $v_0 = 0$, and the attraction of each item is correlated with its revenue through the function $v(r) = \delta^{-\frac{1}{r}}$, where $\delta \in (0,1)$ is a small constant, namely, $v_i = v(r_i)$ for all $i \in [n]$.

A critical observation regarding the above constructed instance is that since $v_0 = 0$, then for any non-empty assortment $S$, $f(S)$ is always a convex combination of the revenues of the items in $S$, for any set $S \subseteq N$ (including the maximizer). Hence, the optimal (unconstrained or cardinality-constrained) solution is always to collect the item with highest revenue, yielding an expected total revenue equal to $\E{\max_{i \in [n]} \frac{v_i}{v_i} r_i} = \E{\max_{i \in [n]} w_i}$. Let $A$ be the given $\rho$-competitive policy for sequential assortment selection -- for simplicity, let us also denote by $A \subseteq N$ the (random) assortment collected by this policy. By applying this policy to the instance constructed above, it must be that 
$$
\E{f(A)} \geq \frac{1}{\rho} \cdot \E{\max_{i \in [n]} w_i}.
$$

We can construct a policy $\Pi$ for the Prophet Inequality problem by simulating policy $A$ on a fictitious instance, as the one described above: initially, for every random variable $w_i$, we define a corresponding item $i$ with parameter distribution $\mathcal{D}_i$, which satisfies $r_i \sim w_i$ and $v_i = \delta^{-\frac{1}{r_i}}$, for some fixed $\delta \in (0,1)$. We provide these distributions to policy $A$ and inform it that $v_0 = 0$. Then, in the online phase, we mimic the decisions of policy $A$, by feeding it the items in the same order as the corresponding arriving rewards. Critically, in order to respect the constraints of the Prophet Inequality instance, once we accept the first item $i$ whose corresponding revenue $r_i$ is collected by $A$, then we reject every remaining item. Note that we can assume without loss of generality that policy $A$ collects at least one item (since, in the opposite case, we can adapt $A$ to collect the last arriving item).

To complete the proof, it suffices to show that for any $\epsilon > 0$, there exists some $\delta > 0$, such that it holds $\E{w_{\iota}} \geq \E{f(A)} -\epsilon$, where $\iota \in N$ is the item collected by policy $\Pi$. Since $\Pi$ collects only the first of the items collected in $A$, it has to be that $\E{w_{\iota}} \geq \E{\min_{i \in A}w_{i}} = \E{\min_{i \in A}r_{i}}$.

For any reward realization and collected set $A \subseteq N$, by assuming that the realized rewards are distinct and taking the limit of $\delta \to 0$, we get
\begin{align*}
    \sum_{i \in A} \frac{v_i}{v(A)} r_i = \sum_{i \in A} \frac{v_i}{v_i + v(S \setminus \{i\})} r_i = \sum_{i \in A} \frac{1}{1 + \sum_{j \in A \setminus \{i\}} \delta^{\frac{1}{r_i}-\frac{1}{r_j}} } r_i \xrightarrow[\delta \to 0]{} \min_{i \in A} r_i,
\end{align*}
where the last equality follows by the fact that the limit of $\sum_{j \in A \setminus \{i\}} \delta^{\frac{1}{r_i}-\frac{1}{r_j}}$ becomes $0$, when $r_i < r_j$ for each $j \in A \setminus \{i\}$, and diverges, otherwise. 

By continuity, for any realization $\theta$ there exists some $\delta(\theta) \in (0,1)$, such that for every instance with $\delta \leq \delta(\theta)$ it holds $\min_{i \in A} r_i \geq \sum_{i \in A} \frac{v_ir_i}{v(A)} - \epsilon$. Notice that the realizations where $\exists i,j \in N \text{ such that }i\neq j \text{ and }r_i = r_j$ have zero measure, since we assume without loss of generality that the reward distributions are continuous. Therefore, by using $\delta = \inf_{\theta} \delta(\theta)$, we get that  
\begin{align*}
    \E{w_{\iota}} \geq \E{\min_{i \in A}w_{i}} = \E{\min_{i \in A}r_{i}} \geq \E{\sum_{i \in A} \frac{v_ir_i}{v(A)}} - \epsilon = \E{f(A)} - \epsilon \geq \frac{1}{\rho} \cdot \E{\min_{i \in A}r_{i}} - \epsilon.
\end{align*}
Given that the above holds for any $\epsilon > 0$, this concludes the proof of the reduction.
\end{proof}

\restatelowerboundstrong*
\begin{proof}
We consider an instance of the problem under MNL demand with two items $i_1$ and $i_2$, arriving in this order. Let us fix some $\delta \in (0,1)$ and $\kappa \in (0,1)$. The revenue and attraction of item $i_1$ are $r_1 = 1$ and $v_1 = \frac{1}{\delta}$. The attraction of the second item is $v_2 = 1$, while its revenue is $r_2 = r = \frac{1}{\delta} \cdot \left(\frac{v_0}{v_2} + 1\right) = \frac{1}{\delta} \cdot \left(v_0 + 1\right)$, with probability $\delta$, and $r_2 = 0$, otherwise. The attraction of the outside option satisfies $v_0 = \frac{1-\kappa}{\kappa} v_1 = \frac{1-\kappa}{\kappa \delta}$.

It is easy to see that the optimal solution $S^*$ in the above instance is defined as follows: if $r_2 = 0$ (with probability $1-\delta$), then $S^* = \{i_1\}$ and $f(\{i_1\}) = \frac{v_1 \cdot r_1}{v_0 + v_1} = \kappa$. In the case where $r_2 = r$ (with probability $\delta$), the optimal solution is $S^* = \{i_2\}$ and $f(\{i_2\}) = \frac{1}{\delta}$. Indeed, in that case $f(\{i_2\}) = \frac{1}{\delta} > 1 \geq \kappa = f(\{i_1\})$, and further 
$$
f(\{i_1,i_2\}) = \frac{v_1 r_1 + v_2 r_2}{v_0 + v_1 + v_2} = \frac{v_1 + v_2 r}{v_0 + v_1 + v_2} = \frac{v_1 + \frac{1}{\delta}(v_0 + v_2)}{v_0 + v_1 + v_2} < \frac{1}{\delta} = f(\{i_2\}).
$$
Thus, for the expected optimal assortment we have $$
\E{f(S^*)} = (1-\delta) \kappa + \delta \frac{1}{\delta} = 1+ (1-\delta) \kappa \xrightarrow[\delta \to 0]{} 1 + \kappa.
$$ 
In addition, the expected probability of purchasing (any item) in the optimal assortment satisfies
\begin{align*}
\gamma = \E{\phi(S^*,S^*)} = (1-\delta) \frac{u_1}{u_0 + u_1} + \delta \frac{u_2}{u_0 + u_2} = (1-\delta) \kappa + \frac{\kappa \delta^2}{{1-\kappa} + \kappa \delta} \xrightarrow[\delta \to 0]{} \kappa.
\end{align*}

After observing item $i_1$ (which is already deterministic), any policy essentially has two options: (i) to reject $i_1$ and collect $i_2$ (only if $r_2 > 0$), or (ii) to collect $i_1$ and then collect $i_2$ only if $r_2 > 0$. Let $A_1$ and $A_2$ denote the (random) subset collected in each case. Clearly, in the first case we have that $\E{f(A_1)} = 1$, since the policy collects $f(\{i_2\}) = \frac{1}{\delta}$ with probability $\delta$. In the second case, we have
$$
\E{f(A_2)} = (1-\delta) \kappa + \delta \frac{v_1 + v_2 r}{v_0 + v_1 + v_2} = (1-\delta) \kappa + \delta \frac{v_1 + \frac{1}{\delta}(v_0 + v_2)}{v_0 + v_1 + v_2} = (1-\delta) \kappa + \frac{\delta v_1 + v_0 + v_2}{v_0 + v_1 + v_2} 
$$
By replacing the values of $v_1,v_2$ and $v_0$, the above becomes
$$
\E{f(A_2)} = (1-\delta) \kappa + \frac{\frac{1-\kappa}{\kappa} \frac{1}{\delta} + 2}{\frac{1-\kappa}{\kappa} \frac{1}{\delta} + \frac{1}{\delta} + 1} = (1-\delta) \kappa + \frac{{1-\kappa} + 2 \kappa \delta}{1 + \kappa\delta} \xrightarrow[\delta \to 0]{} 1.
$$

By continuity, for any $\epsilon' \in (0,1)$ there exists $\delta(\gamma,\kappa,\epsilon') \in (0,1)$, such that for any $\delta < \delta(\gamma,\kappa,\epsilon')$, it holds
\begin{align*}
    |\gamma - \kappa| < \frac{\epsilon'}{2} \text{   ,   } |\E{f(A_2)} - 1| < \epsilon' \text{   and   } |\E{f(S^*)} - 1 - \kappa | \leq \frac{\epsilon'}{2}. 
\end{align*}
Therefore, the best achievable competitive guarantee in the above instance can be lower bounded as
$$
\rho = \frac{\E{f(S^*)}}{\max\left(\E{f(A_1)},\E{f(A_2)}\right)} \geq \frac{1+\kappa-\frac{\epsilon'}{2}}{1+\epsilon'} \geq \frac{1+\gamma-\epsilon'}{1+\epsilon'} = 1+\gamma-\epsilon,
$$
where in the last equality we set $\epsilon' = \frac{\epsilon}{2+\gamma-\epsilon}$. This concludes the proof. 

We remark that in the above proof the attractions of all items are deterministic. Alternatively, one could replace the parameters of the second item with: $r_2 = \frac{1}{\delta} \cdot \left(v_0 + 1\right)$ (deterministically) and $v_2 = 1$, with probability $\delta$, and $v_2 = 0$, otherwise. The resulting instance has deterministic revenues and random attractions and is effectively equivalent for the purpose of this lower bound.
\end{proof}

\section{Appendix: Properties of Demand Models}\label{app:demand}

\restateMNLfacts*

\begin{proof}
By observing the choice probabilities for every assortment $S$ and item $i \in S$ under the MNL choice model, that is $\phi(i,S) = \frac{v_i}{v_0 + v(S)}$, it is easy to verify that the model is substitutable and that \Cref{cond:1} is trivially satisfied, since $\phi(i,S)$ only depends on the realized attractions in $S$, given that $v_0$ is deterministic. To see that \Cref{cond:2} holds, we observe that for any $i \in S$:
$$
\phi(i,S) = \frac{v_i}{v_0 + v(S)} = \frac{v_i}{v_0 + v_i} \cdot \frac{v_0 + v_i}{v_0 + v_i + v(S \setminus \{i\})} \geq \frac{v_i}{v_0 + v_i} \cdot \frac{v_0}{v_0 + v(S \setminus \{i\})} = \phi(i,i) \cdot \phi(0,S \setminus \{i\}),
$$
where the inequality is due to the fact that the function $h(x) = \frac{x}{c+x}$ is non-decreasing in $x\geq 0$, for any $c \geq 0$. Notice that in the case where either $v_i$ or $v_0$ is zero, then \Cref{cond:2} follows trivially. Finally, the MNL choice model satisfies (the strong version of) \Cref{cond:3}, given that $\frac{\phi(0,S)}{\phi(i,S)} = \frac{v_0}{v_i}$ for any $S \subseteq N$.
\end{proof}

\restateRUM*
\begin{proof}
It is well-known that any instance of RUM is substitutable. Indeed, for $S \subseteq T \subseteq N$ and $i \in S$,
\begin{align*}
\phi(i,S) = \Pro{u_i \geq u_j, \forall j \in S \cup \{0\} \setminus \{i\}} \geq \Pro{u_i \geq u_j, \forall j \in T \cup \{0\} \setminus \{i\}} = \phi(i,T).
\end{align*}
\Cref{cond:1} follows trivially by the fact that $\phi(i,S)$ can be evaluated given knowledge of the utility distribution of each item in $S$ (which is encoded in the demand parameters). To prove \Cref{cond:2}, for any $i\in S$, we have
\begin{align*}
\phi(i,S) = \Pro{u_i \geq u_j, \forall j \in S \cup \{0\} \setminus \{i\}} &\geq \Pro{\{u_i \geq u_0\} \wedge \{u_0 \leq u_j, \forall j \in S \setminus \{i\}\}} \\
&= \Pro{u_i \geq u_0} \cdot \Pro{u_0 \leq u_j, \forall j \in S \setminus \{i\}} \\
&= \phi(i,\{i\}) \cdot \phi(0,S\setminus \{i\}),
\end{align*}
where the second equality follows by independence and fact that $v_0$ is deterministic. 
To see that the weak version of \Cref{cond:3} holds, observe that for any set $S \subseteq N$ and $i \in S$, we have
$$
\frac{\phi(0,S)}{\phi(i,S)} \leq \frac{\phi(0,S)}{\phi(i,\{i\}) \cdot \phi(0,S \setminus\{i\})} = \frac{\phi(0,\{i\}) \cdot \phi(0,S \setminus \{i\})}{\phi(i,\{i\}) \cdot \phi(0,S \setminus\{i\})} = \frac{\phi(0,\{i\})}{\phi(i,\{i\})},
$$
where the inequality follows from \Cref{cond:2} and the first equality from noticing that $\phi(0,S) = \phi(0,\{i\}) \cdot \phi(0,S \setminus \{i\})$, using the fact that $u_0$ is deterministic.
\end{proof}

\restateGAM*
\begin{proof}
The GAM is substitutable, since for any $S \subseteq T \subseteq N$ and $i \in S$, we have
\begin{align*}
\phi(i,S) = \frac{v_i}{v_0 + v(S) + w(N \setminus S)} = \frac{v_i}{v_0 + v(S) + w(T\setminus S) + w(N \setminus T)} \geq \frac{v_i}{v_0 + v(T) + w(N \setminus T)} = \phi(i,T),
\end{align*}
where we use the fact that $w_i \leq v_i$ for every $i \in N$.

\Cref{cond:1} follows since $\phi(i,S)$ can be evaluated given knowledge of the the attractions of the items in $S$, given that the shadow utilities and the utility of the outside option are deterministic. 

In order to show \Cref{cond:2}, for any $i\in S$ we can see that 
\begin{align*}
\phi(i,S) = \frac{v_i}{v_0 + v(S) + w(N \setminus S)} &= \frac{v_i}{v_0 + v_i + w(N \setminus \{i\})} \cdot \frac{v_0 + v_i + w(N \setminus \{i\})}{v_0 + v(S) + w(N \setminus S)} \\
&= \phi(i,\{i\}) \cdot \frac{v_0 + v_i + w(N \setminus \{i\})}{v_0 + v_i + v(S\setminus \{i\}) + w(N \setminus S)} \\
&\geq \phi(i,\{i\}) \cdot \frac{v_0 + w(N)}{v_0 + v(S\setminus \{i\}) + w(N \setminus (S\setminus \{i\}) )}\\
&\geq \phi(i,\{i\}) \cdot \frac{v_0 + w(N\setminus (S\setminus \{i\}))}{v_0 + b_i + v(S\setminus \{i\}) + w(N \setminus (S\setminus \{i\}) )}\\
& = \phi(i,\{i\}) \cdot \phi(0,S\setminus \{i\}),
\end{align*}
where the inequality follows by the fact that $w_i \leq v_i$ and that the function $h(x) = \frac{c_1 +x}{c_2 +x}$ is non-decreasing in $x$ for any $c_2 \geq c_1 \geq 0$. The second inequality simply follows from non-negativity of the shadow attractions.

In order to verify the weak version of \Cref{cond:3}, for any set $S \subseteq N$ and $i \in S$, we have
$$
\frac{\phi(0,S)}{\phi(i,S)} = \frac{v_0 + w(N \setminus S)}{v_i} \leq \frac{v_0 + w(N \setminus \{i\})}{v_i} = \frac{\phi(0,\{i\})}{\phi(i,\{i\})}.
$$
\end{proof}

\restateLCF*
\begin{proof}
For proving substitutability, for any $S \subseteq T \subseteq N$ and $i \in S$, we simply observe
$$
\phi(i,S) = q_i \cdot \prod_{\substack{j \in S \\ f_j < f_i}} \left(1 - q_j\right) \geq q_i \cdot \prod_{\substack{j \in S \\ f_j < f_i}} \left(1 - q_j\right) = \phi(i,T),
$$
while for \Cref{cond:2}, we have
$$
\phi(i,S) = q_i \cdot \prod_{\substack{j \in S \\ f_j < f_i}} \left(1 - q_j\right) = \phi(i,\{i\}) \cdot \prod_{\substack{j \in S \setminus \{i\} \\ f_j < f_i}} \left(1 - q_j\right) \geq \phi(i,\{i\}) \cdot \prod_{\substack{j \in S \setminus \{i\}}} \left(1 - q_j\right) = \phi(i,\{i\}) \cdot \phi(0,S\setminus \{i\}).
$$
\Cref{cond:1} can be trivially verified since $\phi(i,S)$ only depends on the parameters of the items in $S$. Finally, to show the weak version of \Cref{cond:3}, we observe that for any $S \subseteq N$ and $i \in S$, it holds
\begin{align*}
    \frac{\phi(0,S)}{\phi(i,S)} = \frac{\prod_{j \in S} \left(1 - q_j\right)}{q_i \cdot \prod_{\substack{j \in S \\ f_j < f_i}} \left(1 - q_j\right)} = \frac{1-q_i}{q_i}  \cdot\frac{\prod_{j \in S \setminus \{i\}} \left(1 - q_j\right)}{\prod_{\substack{j \in S \setminus \{i\} \\ f_j < f_i}} \left(1 - q_j\right)} \leq \frac{1-q_i}{q_i} = \frac{\phi(0,\{i\})}{\phi(i,\{i\})}.
\end{align*}
\end{proof}

\section{Appendix: Efficient Implementation and Limited Information}\label{app:implementation}

In this section, we briefly discuss how our policies can be implemented efficiently in terms of both computational and sample complexity. Let us start with describing a consistency property satisfied by our policies: consider any (unconstrained or constrained) threshold-based $\rho$-competitive policy, among the ones we present in \Cref{sec:unconstrained,sec:cardinality,sec:knapsack}, and let $\tau$ be the threshold that provides the guarantee. By using instead threshold $\tau' \in [\frac{\tau}{\alpha}, \tau]$ for some $\alpha \in [1, \infty)$, then the competitive guarantee of the policy becomes at most $\alpha \rho$. In order to see this, recall that each one of our proofs culminates into an inequality of the form
\begin{align*}
\Ex{}{f(A_\tau)} \geq \lambda \cdot \left(\Ex{}{f(S^*)} - (\rho-1) \cdot \tau\right) + (1-\lambda) \cdot \tau,
\end{align*}
for some constant $\lambda \in [0,1]$. The constant $\rho \in [1,\infty)$ is the resulting competitive guarantee after setting threshold $\tau = \frac{1}{\rho} \cdot \E{f(S^*)}$. By setting instead any threshold $\tau' \in [\frac{\tau}{\alpha}, \tau] $ with $\tau = \frac{1}{\rho} \cdot \E{f(S^*)}$, we get 
\begin{align*}
\Ex{}{f(A_{\tau'})} &\geq \lambda \cdot \left(\Ex{}{f(S^*)} - (\rho-1) \cdot \tau'\right) + (1-\lambda) \cdot \tau' \\
&\geq \lambda \cdot \left(\Ex{}{f(S^*)} - (\rho-1) \cdot \tau\right) + (1-\lambda) \cdot \frac{\tau}{\alpha}\\
&=\left( \lambda \cdot \frac{1}{\rho} + (1-\lambda) \cdot \frac{1}{\alpha\rho}\right) \cdot \Ex{}{f(S^*)} \\
&\geq \frac{1}{\alpha \rho} \cdot \Ex{}{f(S^*)}.
\end{align*}
It is not hard to show that above property also extends to the case of \Cref{alg:budgetfull} for knapsack-constraints.

\paragraph{NP-hard offline problems.} For demand models where the offline assortment selection problem is NP-hard, the above property provides a way to leverage any (deterministic) $\alpha$-approximation algorithm in order to compute the threshold. For any such algorithm, let $S^*_{\alpha}$ be the $\alpha$-approximate feasible solution returned for a given instance, that is, $f(S^*) \geq f(S^*_{\alpha})\geq \frac{1}{\alpha} \cdot f(S^*)$. Hence, one can use this algorithm to compute an approximate threshold of the form $\tau = c \cdot \E{f(S^*)}$ for some $c \in (0,+ \infty)$ by some a threshold $\tau' = c \cdot \E{f(S_{\alpha}^*)} \in [\frac{\tau}{\alpha}, \tau]$. Therefore, by applying any of our $\rho$-competitive policies with threshold $\tau'$, one immediately gets an $\alpha \rho$-competitive policy for the sequential assortment selection problem for any demand model which satisfies our sufficient conditions.

\paragraph{Limited information.} In the regime where the decision-maker has limited information on the priors, it is possible to estimate the thresholds of our policies through sampling. Specifically, by drawing a fresh realization for each item and and computing the optimal solution over the resulting instance, one can get an unbiased sample from the distribution of $f(S^*)$. Therefore, by taking the average of these samples and using standard concentration bounds (and standard mild assumptions), for any $\epsilon,\delta \in (0,1)$, one can compute a value $V$, such that $(1-\epsilon) \E{f(S^*)} \leq V \leq (1 + \epsilon) \E{f(S^*)}$, with probability $1-\delta$. Therefore, by setting $\frac{V}{1+\epsilon}$ instead of $\E{f(S^*)}$ in the definition of our thresholds, the resulting threshold $\tau'$ lies within $[\frac{1-\epsilon}{1+\epsilon}\cdot \tau,\tau]$. This immediately yields a $\frac{1+\epsilon}{1-\epsilon} \rho$-competitive policy (with probability $1-\delta$) for all the demand models and constraints considered in this work, where $\rho$ is the competitive guarantee under complete knowledge of the distributions.

\section{Appendix: Prophet Inequality for Convex Combinations}
\label{app:alternative}

We consider the case of discrete choice models where for any assortment $S \subseteq N$, it holds $\psi(S) = \phi(S,S) = 1$ (equivalently $\phi(0,S) = 0$), namely, situations where the buyer always purchases one of the offered items. A crucial observation in this case is that for any set $S \subseteq {N}$, it holds $f(S) \in \text{conv}(\{r_i \mid \forall i \in S\})$, namely, $f(S)$ is a convex combination of the revenues in $S$, each weighted by $\phi(i,S)$. Hence, the optimal set for any realization is a singleton and corresponds to the item with highest revenue, that is, $f(S^*) = \max_{i \in N} r_i$. It is easy to see that the unconstrained problem of sequential assortment selection in this case can be reduced to an instance of the original Prophet Inequality, by (fictitiously) restricting the decision-maker to collect at most one item. This reduction already yields $2$-competitive policies by choosing the first item (if any), whose revenue surpasses either threshold $\tau = \frac{1}{2} \E{\max_{i \in N} r_i}$, or $\tau = \text{median}(\max_{i \in N} r_i)$ \cite{KW12,S84}. 

We show that the above threshold-based policies maintain their competitive guarantees, even if the decision-maker can collect more than one (threshold-feasible) items. In fact, by studying a ``proxy'' variant of the Prophet Inequality problem, we prove that, whenever a discrete choice model satisfies $\phi(0,S) = 0$ for any $S \subseteq N$, then the rest of its properties are irrelevant.

\paragraph{Prophet Inequality for convex combinations.} We consider a variant of the Prophet Inequality problem where the decision-maker observes a sequence of $n$ random non-negative rewards $w_1, \ldots, w_n$ -- each drawn independently from a known distribution -- and is allowed to collect any of them (without any constraints). At the end of the process, an adversary who observes all the realized rewards, chooses a value $\alpha_i \in [0,1]$ for each $i \in A$ such that $\sum_{i \in A} \alpha_i = 1$, and the decision-maker finally collects $f(A) = \sum_{i \in A} \alpha_i \cdot w_i$, or $f(\emptyset) = 0$ in the case where $A = \emptyset$. 

Notice, that the expected optimal reward in the above setting is given by $\E{\max_{i \in N} w_i}$. Further, assuming that the adversary chooses a worst-case convex combination of the collected rewards, the decision-maker collects $f(A) = \min_{i \in A} w_i$.

Clearly, the above setting subsumes the problem of unconstrained sequential assortment selection for any demand model that gives zero probability to the outside option, since the adversary can simulate the effect of any such model by setting $\alpha_i = \phi(i,A)$ for any $i \in A$ (here, every reward $w_i$ follows the marginal distribution of the corresponding revenue $r_i$). Notice, further, that the above setting abstracts the original Prophet Inequality in the regime of fixed-threshold policies and almighty adversaries. Indeed, by choosing a worst-case convex combination, the adversary can simulate the worst-case arrival order of the items in the latter (that is, to reveal the items in non-decreasing order of reward).

The following result can be proved by slightly adapting the most ``modern'' proof of the original Prophet Inequality result, due to Kleinberg and Weinberg \cite{KW12}. For simplicity, we assume without loss of generality that all the reward distributions are continuous. 

\begin{theorem} By accepting all items with realized reward greater or equal to $\tau = \frac{1}{2} \E{\max_{i \in N} w_i}$ (or, equivalently, $\tau = \textrm{median}(\max_{i \in N} w_i)$), we get an optimal $2$-competitive prophet inequality for convex combinations.
\end{theorem}
\begin{proof}
Let $\tau > 0$ be any fixed threshold and let $A_{\tau} = \{i \in N \mid w_i \geq \tau\}$ be the set of collected items. By using the non-negativity of $f(A_{\tau})$, we can write
\begin{align*}
    \Ex{}{f(A_{\tau})} = \int^{\infty}_{x = 0} \Pro{f(A_{\tau}) > x} \text{d}x = \int^{\tau}_{x = 0} \Pro{f(A_{\tau}) > x}\text{d}x + \int^{\infty}_{x = \tau} \Pro{f(A_{\tau}) > x}\text{d}x.
\end{align*}
Let $p = \Pro{\max_{i \in N} w_i \geq \tau}$ and notice that for any $x \leq \tau$, it holds 
$$\Pro{f(A_{\tau}) > x} \geq \Pro{f(A_{\tau}) \geq \tau} = p,$$
where the equality follows from a simple observation: if $\max_{i \in N} w_i \geq \tau$, then there exists a convex combination of rewards at least $\tau$ (with value at least $\tau$), and the opposite. 

Fix now any $x \geq \tau$. By considering all possible non-trivial partitions of $N$ into $A_{\tau}$ and $N \setminus A_{\tau}$, we have
\begin{align*}
    \Pro{f(A_{\tau}) > x} &= \sum_{S \subseteq {N}, S\neq \emptyset} \Pro{\big\{f(S)>x \big\} \wedge \big\{w_i \geq \tau,~ \forall i \in S \big\} \wedge \big\{w_i < \tau,~ \forall i \in {N} \setminus S\big\}}\\
    &= \sum_{S \subseteq {N}, S\neq \emptyset} \Pro{\big\{f(S)>x \big\} \wedge \big\{w_i \geq \tau,~ \forall i \in S \big\}} \cdot \Pro{w_i < \tau,~ \forall i \in {N} \setminus S} \\
    &\geq (1-p) \cdot \sum_{S \subseteq {N}, S\neq \emptyset} \Pro{\big\{f(S)>x \big\} \wedge \big\{w_i \geq \tau,~ \forall i \in S \big\}} \\
    &\geq (1-p)\cdot \Pro{\max_{i \in N} w_i > x},
\end{align*}
where the second equality above follows from the fact that the rewards of different items are independent, and the first inequality by the fact that $\Pro{w_i < \tau,~\forall i \in N \setminus S} \geq \Pro{w_i < \tau,~\forall i \in N} = 1-p$. The last equality follows from a simple union bound:
\begin{align*}
\sum_{S \subseteq {N}, S\neq \emptyset} \Pro{\big\{f(S)>x \big\} \wedge \big\{w_i \geq \tau,~ \forall i \in S \big\}} \geq \Pro{\exists S \subseteq N, S \neq \emptyset\text{ s.t. }f(S)>x} \geq \Pro{\max_{i \in N} w_i > x},
\end{align*}
since $f(S)$ is always a convex combination of the items in $S$.

By combining the above, we can conclude that
\begin{align*}
    \Ex{}{f(A_{\tau})} \geq p \cdot \tau + (1-p) \cdot \int^{\infty}_{x = \tau} \Pro{\max_{i \in N} w_i > x}\text{d}x \geq p \cdot \tau + (1-p) \cdot \left(\E{\max_{i \in N} w_i } - \tau\right).
\end{align*}
The proof is concluded by verifying that both choices of thresholds imply $2$-competitive policies. 
\end{proof}
\end{document}